\documentclass[11pt]{article}
\pdfoutput=1
\usepackage[T1]{fontenc}
\usepackage{lmodern}
\usepackage[protrusion=true,expansion=true]{microtype}
\usepackage{amsmath,amssymb,amsfonts,amsthm}
\usepackage{graphicx}
\usepackage{setspace}
\usepackage{fullpage}
\usepackage[backref=page]{hyperref}
\usepackage{color}
\usepackage{wrapfig}
\usepackage{tikz}
\usetikzlibrary{decorations.pathreplacing}
\usepackage{algorithm}
\usepackage[noend]{algpseudocode}
\usepackage[framemethod=tikz]{mdframed}
\usepackage{xspace}
\usepackage{pgfplots}
\usepackage{framed}
\usepackage{subcaption}
\usepackage{thmtools}
\usepackage{thm-restate}
\pgfplotsset{compat=1.5}

\newtheorem{theorem}{Theorem}[section]
\newtheorem{corollary}[theorem]{Corollary}
\newtheorem{lemma}[theorem]{Lemma}

\newtheorem{definition}[theorem]{Definition}

\newenvironment{proofof}[1]{\begin{trivlist} \item {\bf Proof
#1:~~}}
  {\qed\end{trivlist}}
\renewenvironment{proofof}[1]{\par\medskip\noindent{\bf Proof of #1: \ }}{\hfill$\Box$\par\medskip}
\newcommand{\namedref}[2]{\hyperref[#2]{#1~\ref*{#2}}}
\newcommand{\thmlab}[1]{\label{thm:#1}}
\newcommand{\thmref}[1]{\namedref{Theorem}{thm:#1}}
\newcommand{\lemlab}[1]{\label{lem:#1}}
\newcommand{\lemref}[1]{\namedref{Lemma}{lem:#1}}

\newcommand{\figlab}[1]{\label{fig:#1}}
\newcommand{\figref}[1]{\namedref{Figure}{fig:#1}}
\newcommand{\alglab}[1]{\label{alg:#1}}
\renewcommand{\algref}[1]{\namedref{Algorithm}{alg:#1}}

\newcommand{\PPr}[1]{\ensuremath{\mathbf{Pr}\left[#1\right]}}
\newcommand{\PPPr}[2]{\ensuremath{\underset{#1}{\mathbf{Pr}}\left[#2\right]}}
\newcommand{\Ex}[1]{\ensuremath{\mathbb{E}\left[#1\right]}}
\newcommand{\EEx}[2]{\ensuremath{\underset{#1}{\mathbb{E}}\left[#2\right]}}
\renewcommand{\O}[1]{\ensuremath{\mathcal{O}\left(#1\right)}}
\newcommand{\tO}[1]{\ensuremath{\tilde{\mathcal{O}}\left(#1\right)}}
\newcommand{\eps}{\varepsilon}
\newcommand{\TVD}{d_{\mathrm{tv}}}
\newcommand{\KLD}{d_{\mathrm{KL}}}

\def \FindColumn    {\mdef{\textsc{FindColumn}}}
\def \fracpart    {\mdef{\textsc{Frac}}}

\def \ba    {\mdef{\mathbf{a}}}
\def \calA    {\mdef{\mathcal{A}}}

\def \calD    {\mdef{\mathcal{D}}}

\def \calE    {\mdef{\mathcal{E}}}

\def \calN    {\mdef{\mathcal{N}}}

\def \calS    {\mdef{\mathcal{S}}}

\def \calU    {\mdef{\mathcal{U}}}

\def \bA    {\mdef{\mathbf{A}}}
\def \bB    {\mdef{\mathbf{B}}}

\def \bD    {\mdef{\mathbf{D}}}

\def \bM    {\mdef{\mathbf{M}}}

\def \bS    {\mdef{\mathbf{S}}}

\def \be    {\mdef{\mathbf{e}}}

\def \bq    {\mdef{\mathbf{q}}}

\def \bu    {\mdef{\mathbf{u}}}
\def \bw    {\mdef{\mathbf{w}}}
\def \bv    {\mdef{\mathbf{v}}}
\def \bx    {\mdef{\mathbf{x}}}
\def \by    {\mdef{\mathbf{y}}}
\def \bz    {\mdef{\mathbf{z}}}

\newcommand{\mdef}[1]{{\ensuremath{#1}}\xspace}  % Math Def which can also be used in normal text.
\DeclareMathOperator*{\polylog}{polylog}
\DeclareMathOperator*{\poly}{poly}

\DeclareMathOperator*{\Trace}{Tr}
\DeclareMathOperator*{\Bern}{Bern}

\newcommand{\ignore}[1]{}

\newif\ifnotes\notestrue %set this to true if notes are visible and to false (next line) if they should be hidden
\ifnotes
\newcommand{\david}[1]{\textcolor{purple}{{\bf (David:} {#1}{\bf ) }} \marginpar{\tiny\bf
             \begin{minipage}[t]{0.5in}
               \raggedright S:
            \end{minipage}}}    

\newcommand{\samson}[1]{\textcolor{blue}{{\bf (Samson:} {#1}{\bf ) }} \marginpar{\tiny\bf
             \begin{minipage}[t]{0.5in}
               \raggedright S:
            \end{minipage}}}            							
\else
\newcommand{\samson}[1]{}
\newcommand{\david}[1]{}
\fi

\makeatletter
\renewcommand*{\@fnsymbol}[1]{\textcolor{mahogany}{\ensuremath{\ifcase#1\or *\or \dagger\or \ddagger\or
 \mathsection\or \triangledown\or \mathparagraph\or \|\or **\or \dagger\dagger
   \or \ddagger\ddagger \else\@ctrerr\fi}}}
\makeatother

\providecommand{\email}[1]{\href{mailto:#1}{\nolinkurl{#1}\xspace}}

\definecolor{mahogany}{rgb}{0.75, 0.25, 0.0}
\definecolor{bleudefrance}{rgb}{0.19, 0.55, 0.91}
\hypersetup{
     colorlinks   = true,
     citecolor    = bleudefrance,
     linkcolor	 = mahogany
}

\AtBeginDocument{%
  \DeclareFontShape{T1}{lmr}{m}{scit}{<->ssub*lmr/m/scsl}{}%
}

\title{A Strong Separation for Adversarially Robust $\ell_0$ Estimation for Linear Sketches}
\author{
Elena Gribelyuk \\ Princeton University \\ \email{eg5539@princeton.edu}  
\and 
Honghao Lin \\ Carnegie Mellon University \\ \email{honghaol@andrew.cmu.edu} 
\and 
David P. Woodruff \\ Carnegie Mellon University \\ \email{dwoodruf@andrew.cmu.edu}
\and
Huacheng Yu \\ Princeton University \\ \email{hy2@cs.princeton.edu}
\and
Samson Zhou \\ Texas A\&M University \\ \email{samsonzhou@gmail.com}
}
\date{\today}

\begin{document}
\allowdisplaybreaks

\maketitle

\begin{abstract}
The majority of streaming problems are defined and analyzed in a static setting, where the data stream is any worst-case sequence of insertions and deletions that is fixed in advance. However, many real-world applications require a more flexible model, where an adaptive adversary may select future stream elements after observing the previous outputs of the algorithm. Over the last few years, there has been increased interest in proving lower bounds for natural problems in the adaptive streaming model. In this work, we give the first known adaptive attack against linear sketches for the well-studied $\ell_0$-estimation problem over turnstile, integer streams. For any linear streaming algorithm $\mathcal{A}$ that uses sketching matrix $\mathbf{A}\in \mathbb{Z}^{r \times n}$ where $n$ is the size of the universe, this attack makes $\tilde{\mathcal{O}}(r^8)$ queries and succeeds with high constant probability in breaking the sketch. We also give an adaptive attack against linear sketches for the $\ell_0$-estimation problem over finite fields $\mathbb{F}_p$, which requires a smaller number of $\tilde{\mathcal{O}}(r^3)$ queries. Finally, we provide an adaptive attack over $\mathbb{R}^n$ against linear sketches $\mathbf{A} \in \mathbb{R}^{r \times n}$ for $\ell_0$-estimation, in the setting where $\mathbf{A}$ has all nonzero subdeterminants at least $\frac{1}{\poly(r)}$. Our results provide an exponential improvement over the previous number of queries known to break an $\ell_0$-estimation sketch.
\end{abstract}

\section{Introduction}
In the classical streaming model, updates to an underlying dataset arrive sequentially and the goal is to compute or approximate some predetermined statistic of the dataset while using space sublinear in $m$, the length of the stream and $n$, the dimension of the underlying dataset; ideally, the algorithm should provide this estimate after making only a single pass over the data. The streaming model of computation captures key memory and resource requirements of algorithms in many big data applications, and has therefore emerged as a central paradigm for applications where the size of the data is significantly larger than the available storage, such as logs generated from either virtual or physical traffic monitoring, stock market transactions, scientific observations, and machine and sensor data, e.g., Internet of Things (IoT) sensors, financial markets, and scientific observations.

Observe that in many of these applications, intermediate outputs of the algorithm may impact the distribution of future inputs to the algorithm. 
For example, in database systems, future queries to the database may be dependent on the full history of responses by the database algorithm to previous queries. In optimization procedures such as stochastic gradient descent, the update at each time step can be based on the history of previous outputs. In recommendation systems, a user may choose to remove some suggestions based on personal preference and then query for a new list of recommendations.
Additionally, statistics aggregated from financial markets on the current day could result in algorithmic decisions that impact certain enterprises, thereby affecting their future evaluations, which form a small but nonzero component of the information collected by the algorithm on the next day. 

Unfortunately, the classical \emph{oblivious} streaming model assumes that the input is fixed in advance to be the worst possible permutation of elements. Moreover, since the algorithm only provides an estimate \textit{once} at the end of the stream, we may assume that the input stream is independent of the internal randomness of the streaming algorithm. Indeed, the analyses of many randomized streaming algorithms crucially utilize the independence between the internal randomness of the algorithm and the data stream. 
However, as discussed previously, this may not be a reasonable assumption for the above applications and many additional settings~\cite{MironovNS11,BogunovicMSC17,NaorY19,AvdiukhinMYZ19,CherapanamjeriN20,CherapanamjeriSWZZZ23,Cohen0NSSS22,CohenNSS23,DinurSWZ23,WoodruffZZ23a, Cohen-Ahmadian24}. This motivates the \textit{adversarially robust} streaming model, which we discuss next.

\paragraph{The adversarially robust streaming model.}
To address these shortcomings of the classical oblivious streaming model, the adversarially robust streaming model was recently proposed~\cite{Ben-EliezerJWY22} to capture settings where the sequence of inputs to the streaming algorithm can be adaptive or even adversarial. At each time $t \in [m]$, the streaming algorithm $\mathcal{A}$ receives an update $u_t = (a_t, \Delta_t)$, where each $a_t \in [n]$ is an index and $\Delta_t \in \mathbb{Z}$ denotes an increment or decrement to index $a_t$ in the underlying frequency vector $\bx$, i.e. the $i^{\textrm{th}}$ coordinate of the frequency vector is given by $\bx_i = \sum_{t: a_t = i} \Delta_t$. Similarly, let $\bx^{(t)}$ denote the state of the frequency vector restricted to the first $t$ updates, i.e. $\bx^{(t)}_i = \sum_{s \leq t: a_s = i} \Delta_s$.
We consider the setting where $m = \textrm{poly}(n)$ and $|\Delta_t| \leq \textrm{poly}(n)$ for all $t \in [m]$. Note that by scaling, we could have also assumed that each $\Delta_t$ is an integer multiple of $\frac{1}{\poly(n)}$. Then, $\mathcal{A}$ is an \textit{adversarially robust} streaming algorithm for some estimation function $g:\mathbb{Z}^n \rightarrow \mathbb{R}$ if $\mathcal{A}$ satisfies the following requirement.

\begin{definition}\cite{Ben-EliezerJWY22}
    Let $g:\mathbb{Z}^n \rightarrow \mathbb{R}$ be a fixed function. Then, for any $\eps > 0$ and $\delta >0$, at each time $t \in [m]$ for $m = \poly(n)$, we require our algorithm $\mathcal{A}$ to return an estimate $z_t$ for $g(\bx^{(t)})$ such that $$\Pr\left[\left|z_t - g(\bx^{(t)})\right| \leq \eps g(\bx^{(t)})\right] \geq 1-\delta$$
\end{definition}

The above definition is also known as the \textit{strong tracking guarantee} for adversarial robustness, as defined in \cite{Ben-EliezerJWY22}. Moreover, we may view the adversarial setting as a two-player game between a randomized streaming algorithm $\mathcal{A}$ and an unbounded adversary. In particular, the adversary aims to construct a hard sequence of adaptive\footnote{Note that we use ``adaptive'' and ``adversarial'' interchangeably: both terms indicate that future updates or queries may depend on previous updates and responses of the algorithm.} updates $\{u_1^*, ..., u_m^*\}$ such that any streaming algorithm $\mathcal{A}$ that produces $(\eps, \delta)$-approximate responses $\{z_t\}_{t = 1}^m$ will fail to estimate $g(\bx_{t^*})$ with constant probability at some step $t^* \in [m]$ during the stream. For a chosen function $g$, the game proceeds as follows:

\begin{enumerate}
    \item In each round $t \in [m]$, the adversary selects an update $u_t$ to append to the stream to implicitly define the underlying dataset $\bx^{(t)}$ at time $t$. Importantly, note that $\bx^{(t)}$ may depend on all previous updates $u_1,..., u_{t-1}$, as well as the corresponding responses $z_1, ..., z_{t-1}$ of the streaming algorithm $\mathcal{A}$.
    \item $\mathcal{A}$ receives update $u_t$ and updates its internal state. 
    \item Then, $\mathcal{A}$ returns an estimate $z_t(\bx^{(t)})$ for $g(\bx^{(t)})$ based on the stream observed until time $t$, and progresses to the next round.
\end{enumerate}

Observe that this sequential game only permits a single pass over the data stream. Alternatively, the adversary may choose to only query the streaming algorithm at specific times during the stream. In future sections, we will let $\bx^{(t)}$ denote the query vector at time $t$, which may have been formed by a sequence of $\O{n}$ insertions or deletions to various indices of the previous query vector $\bx^{(t-1)}$.

\paragraph{Insertion-only streams.}
In the insertion-only streaming model, each update $u_t=(a_t, \Delta_t)$ represents an \textit{insertion} of an element $a_t \in [n]$ into the stream $\Delta_t > 0$ times. This corresponds to incrementing the $(a_t)$-th coordinate of the underlying frequency vector $x_{a_t} = x_{a_t} + \Delta_t$. 
In the special case that the increments $\Delta_t = 1$ in each step $t \in [m]$, the $(a_t)$-th coordinate of $\bx$ is simply the number of times that element $a_t$ appeared in the stream. 

In the adversarially robust streaming model with insertion-only updates, it is known that many central streaming problems admit sublinear space algorithms, c.f.,~\cite{kaplan2020,BravermanHMSSZ21,WoodruffZ21,AjtaiBJSSWZ22,Ben-EliezerJWY22,ChakrabartiGS22,BeimelKMNSS22,JiangPW22,AssadiCGS23, ACSS23}. 
In particular, \cite{BravermanHMSSZ21} showed that by using the popular merge-and-reduce framework, adversarial robustness is essentially built into the analysis for a wide class of problems such as clustering, subspace embeddings, linear regression, and graph sparsification. 
In other words, there exist adversarially robust algorithms for these problems that use the same sampling-based approach as classical streaming algorithms in the case where the inputs must be insertion-only. 
Similarly, \cite{WoodruffZ21} showed that for fundamental problems such as norm and moment estimation, distinct elements estimation, heavy-hitters, and entropy estimation, there exist adversarially robust algorithms that pay a small polylogarithmic overhead over the classical insertion-only streaming algorithms that use sublinear space. 

\paragraph{Turnstile streams.}
There is significantly less known about adversarially robust streaming algorithms with turnstile updates. 
The work of \cite{Ben-EliezerJWY22} gives an algorithm that uses space sublinear in the stream length $m$ in the case that the stream has \textit{bounded deletions}. 
However, in the general turnstile streaming setting, the best known adaptive upper bounds are still much worse than in the oblivious case. 
The work of \cite{kaplan2020} showed a way to use differential privacy to protect the internal randomness of the streaming algorithm from the adversary: this framework converts an oblivious streaming algorithm for estimation problem $f$ into an adversarially robust streaming algorithm for the same problem, with an $\tO{\sqrt{m}}$ blow-up in the space complexity for turnstile streams\footnote{We use the notation $\tO{f}$ to represent $f\cdot\polylog(f)$.}. 
More recently, the work of \cite{Ben-EliezerEO22} gave an adversarially robust streaming algorithm for $F_p$ estimation by combining the differential privacy framework of  \cite{kaplan2020} with standard results from sparse recovery; for $\ell_0$ estimation, this reduced the space blow-up to $\tO{m^{1/3}}$. 
Still, when the stream length $m$ is a sufficiently large polynomial of the dimension $n$ of the frequency vector, the space complexity of the algorithm is not sublinear in $n$. 

A natural question is whether there is an inherent space-complexity separation between oblivious and adaptive streaming. To this end, \cite{HardtW13} showed that no linear sketch can approximate the $\ell_2$ norm within even a polynomial multiplicative factor against adaptive queries when the sketching matrix and input stream are both real-valued. 
First, a natural idea is to try to adapt the attack therein to obtain an attack against linear sketches for $\ell_p$-estimation in the integer setting. 
However, the attack of \cite{HardtW13} crucially relies on Gaussian rotational invariance to argue that the algorithm's observations can be parametrized solely by the norms of the inputs. It is not clear whether it is possible to discretize the Gaussian queries of their attack, as the direction in the sketch space may still reveal some information about the norm. Secondly, we remark that \cite{HardtW13} also cannot handle the case of $\ell_0$-estimation over the reals, since $\ell_0$ is not a norm (since the attack requires $\|Cx\| = C \|x\|$ for scalars $C > 0$). 
Thus, an entirely different approach is needed to handle $\ell_0$-estimation over the integers.

Additionally, in 2021, \cite{SADA} showed that there exists a streaming problem for which there is an exponential separation in the space complexity needed to solve the problem in the oblivious and adaptive settings; specifically, this lower bound is shown for a streaming version of the adaptive data analysis problem in the bounded-storage model of computation. 
Later, the work of \cite{ChakrabartiGS22} noted an elegant quadratic separation between oblivious and adaptive streaming for the minimum spanning forest problem over streams with edge insertions and deletions. 
Subsequently, \cite{ChakrabartiGS22} showed a separation for oblivious and adaptive streams for insertion-only streams for the problem of graph-coloring. 
Thus, a well-known open problem~\cite{SADA,Ben-EliezerJWY22,stoc2021workshop,focs2023workshop} is the following:

\begin{quote}
    {\em Is there a separation between oblivious and adaptive turnstile streaming for any natural ``statistical'' streaming estimation problem?}
\end{quote}
We make progress toward answering this question in the affirmative, as we show a lower bound against linear sketches for $\ell_0$ estimation in the adversarial streaming model.

\paragraph{$\ell_0$-estimation problem and linear sketching in the adversarial streaming model.}
In this work, we study the classical streaming problem of estimating the number of distinct elements in a turnstile stream, also known as the $\ell_0$-estimation problem, where $\|x\|_0 = |\{i : x_i \neq 0\}|$. 
Given a stream of updates $u_1 = (a_1, \Delta_1),..., u_m = (a_m, \Delta_m)$, let $a_t \in [n]$ be an index and let $\Delta_t \in \mathbb{Z}$ denote an increment or decrement to index $a_t$ of the underlying frequency vector $\bx \in \mathbb{Z}^n$, where $|\Delta_t| \leq \poly(n)$. 
The task of the streaming algorithm $\mathcal{A}$ is to produce an estimate $z$ such that $\Pr\left[\left|z - \|\bx\|_0 \right| \leq \eps \|\bx\|_0 \right]\geq 1-\delta$ for any $\eps, \delta > 0$ fixed in advance. 
The $\ell_0$-estimation problem has been studied extensively in the last 40 years, beginning with the seminal work of (Flajolet and Martin, FOCS, 1983) \cite{FMalgorithm}. 
The frequency moment estimation problem has since been studied in many other works \cite{Countingdistinct,IndykWoodruff,KaneNelsonWoodruff,FastMoments}, culminating in a nearly optimal algorithm for $\ell_0$-estimation in turnstile streams of \cite{KaneNelsonWoodruff}, which succeeds with high constant probability and gives a $(1 \pm \eps)$-approximation using $\O{\eps^{-2} \log(n)\left(\log\frac{1}{\eps} + \log \log n\right)}$ bits of space.

Moreover, we focus on the case that $\mathcal{A}$ is a \emph{linear} streaming algorithm, meaning that $\mathcal{A}$ samples a sketching matrix $\bA \sim \mathcal{S}$, and for any input $\bx \in \mathbb{Z}^n$, $\mathcal{A}$ returns $f(\bA, \bA \bx)$, where $f$ is any function. 
It is important to note that all known turnstile streaming algorithms are linear sketches, and in fact, it is known that when the stream length is long enough, turnstile streaming algorithms with fixed inputs $\bx$ can be captured by maintaining a linear sketch $\bA\bx$ over the course of the stream \cite{LiNW14,AiHLW16}. 
Motivated by the reasons above, we focus on proving lower bounds against linear sketches for the $\ell_0$-estimation problem in the adaptive streaming setting. 

\subsection{Our Results}
We resolve the open problem posed above by giving the first known adaptive attack against linear sketches for the turnstile $\ell_0$-estimation problem over the integers. Our results are derived from the following promise problem.

\begin{definition}[$\ell_0$ gap norm problem]
Let $0 \leq \alpha < \beta \leq 1$. 
We say that an algorithm $\mathcal{A}$ solves the $(\alpha, \beta)$-$\ell_0$ gap norm problem if, for any input $x \in \mathbb{Z}^n$, $\mathcal{A}$ outputs $0$ if $\|x\|_0 \leq \alpha n$ and outputs $1$ if $\|x\|_0 \geq \beta n$. If $\|x\|_0$ satisfies neither of these conditions, $\mathcal{A}$ may return either $0$ or $1$.
\end{definition}

Furthermore, we focus our attention on linear streaming algorithms, defined as follows:

\begin{definition}[Linear streaming algorithm]
Let $\mathcal{A}$ be a streaming algorithm for a function $g$, and let $A \in \mathbb{Z}^{r \times n}$ be a sketching matrix of its choice, sampled from some distribution $A \sim \mathcal{S}$ over sketching matrices. We say that a streaming algorithm $\mathcal{A}$ is linear if, for every update $x \in \mathbb{Z}^n$, $\mathcal{A}$ observes $Ax$ and returns an estimate $f(A, Ax)$, where $f$ is any function. 
\end{definition}

In all of our results, we assume that the dimensions of the sketching matrix $\bA$ satisfy $r \ll n$.

\begin{theorem}[Informal version of \thmref{thm:main-theorem}]
Suppose that $\mathcal{A}$ is a linear streaming algorithm that solves the $(\alpha + c, \beta - c)$- $\ell_0$ gap norm problem for some constants $\alpha, \beta, c$.
Then there exists a randomized adversary that, with high constant probability can generate a distribution $D$ over $\mathbb{Z}^n$ such that $\mathcal{A}$ fails on $D$ with constant probability. 
This adaptive attack makes at most $\tO{r^{8}}$ queries and runs in $\poly(r)$ time. 
\end{theorem}

This result has implications beyond adversarial streaming. 
In particular, since the existence of a so-called pseudodeterministic streaming algorithm for a particular task implies the existence of adversarially robust streaming algorithm for the same task, our attack implies that any linear pseudodeterministic algorithm for the turnstile $\ell_0$-estimation over the integers can be made to fail after $\poly(r)$ adaptive queries. 
This relates to open questions raised in \cite{PDstreaming}, which asked whether there can be linear pseudodeterministic streaming algorithms for the $\ell_2$ estimation problem. 

Next, we give an attack against linear sketches for $\ell_0$-estimation where all entries of the sketching matrix $\bA$ and input $\bx$ are over $\mathbb{F}_p$ for some prime $p$. We note that known $\ell_0$ sketches can also be adapted to work over such fields with minimal changes (see, e.g., footnote 2 of \cite{mcgregor2011polynomial}).

\begin{theorem}[Informal version of \thmref{finitefield}]
Suppose $\mathcal{A}$ is a linear streaming algorithm that solves the $(\alpha + c, \beta - c)-\ell_0$ gap norm problem with some constants $\alpha$, $\beta$, and $c$. 
There exists an adaptive attack that makes $\tO{r^3}$ queries and with high constant probability outputs a distribution $D$ over $\mathbb{Z}^{n}$ such that when $\bx \sim D$, $\calA$ fails to decide the $\ell_0$ gap norm problem with constant probability.
\end{theorem}

Finally, we give an attack against linear sketches with real entries in the case that sketching matrix $\bA \in \mathbb{R}^{r \times n}$ has all nonzero subdeterminants at least $\frac{1}{\poly(r)}$. 
We remark that this is a natural class of sketching matrices to consider, as the known sketches have this property. 

\begin{theorem}[Informal version of \thmref{thm:real}]
Suppose that $\mathcal{A}$ is a linear streaming algorithm that solves the $(\alpha + c, \beta - c)$-$\ell_0$ gap norm problem with some constants $\alpha, \beta$ and $c$, where $\bA \in \mathbb{R}^{r \times n}$ is the sketching matrix such that $\bA$ has all nonzero subdeterminants at least $\frac{1}{\poly(r)}$. 
Then there exists a randomized algorithm, which after making an adaptive sequence of queries to $f(\bA, \bA \bx)$, with high constant probability can generate a distribution $D$ on $\mathbb{R}^n$ such that $f(\bA, \bA \bx)$ fails on $D$ with constant probability. Moreover, this adaptive attack algorithm makes at most $\poly(r)$ queries and runs in $\poly(r)$ time.
\end{theorem}

This attack serves as a proof-of-concept and as further motivation for our fingerprinting-based techniques. 
Additionally, in a recent work on adversarially-robust property-preserving hash functions \cite{l0conjecture}, it was conjectured that there is an efficient adaptive attack against linear sketches for $\ell_0$-estimation over the reals; our attack resolves this question for the class of sketching matrices with not-too-small subdeterminants.

\subsection{Technical Overview}
In this section, we give a description of the attack against linear sketches for the $\ell_0$-estimation problem.

As the ``adaptive adversary'', the primary goal of our attack is to gradually learn the sketching matrix $\bA$, and design ``harder'' queries as more of $\bA$ becomes known to us.
A sketching matrix $\bA$ may preserve a ``significant amount of information'' about some coordinates $\bx_i$ in $\bA \bx$ (e.g., when there is a row of $\bA$ that is nonzero only in column $i$, $\bA \bx$ can recover $\bx_i$ precisely), while it only mildly ``depends on'' the other coordinates (e.g., when a coordinate $i$ is always ``mixed'' in a sum of many coordinates).
The coordinates that $\bA$ preserves a significant information about, or the \emph{significant coordinates}, can be very useful for estimating the $\ell_0$-norm when the queries are non-adaptive. 
For example, one may sample $\bA$ in a careful way such that a random set of $\O{1}$ coordinates is significant, and from $\bA \bx$, one can approximately identify whether each of them is zero. Then, just based on the fraction of non-zeroes among these sampled coordinates, the $\ell_0$-norm can already be approximated up to an additive error of, say $0.1n$, solving $\ell_0$ gap norm.

Thus, our main strategy is to gradually identify the significant coordinates, and set them to zero in all future queries as soon as we find any.\footnote{In fact, zeroing out these coordinates after we learn them is the only type of adaptive move in our attack.}
This makes the future queries \emph{harder for $\bA$}, since intuitively, $\bA$ would be wasting some of its budget on a coordinate that is always zero, effectively reducing its dimension.
When the number of rows $r\ll n$, $\bA$ cannot simultaneously preserve a significant amount of information for too many $\bx_i$'s.
After we have learned all such coordinates, the query algorithm would have to only rely on the other \emph{insignificant coordinates}, which $\bA \bx$ only mildly depends on.

In order to perform such an attack, there are three main problems to solve:
\begin{itemize}
    \item define ``significance'' and show that not too many coordinates are significant when $r\ll n$;
    \item show that we can learn which coordinates are significant using polynomially many queries;
    \item show that the query algorithm cannot estimate the $\ell_0$-norm accurately when $\bx$ is supported only on the insignificant coordinates.
    In fact, we will design query distributions for $\bx$ with very different $\ell_0$-norms, such that the impact of the insignificant coordinates on the sketch $\bA\bx$ is nearly identical, regardless of the $\ell_0$-norm.
\end{itemize}
In the following, we elaborate on how we solve the problems described above.

\paragraph{Fingerprinting codes.} 
First, let us see how we should learn the significant coordinates.
While we have not formally defined the ``significant coordinates'' yet, for now let us focus on an \emph{important} extreme case: the sketch $\bA \bx$ is simply an (unknown) subset of $r$ coordinates of $\bx$, i.e., each row of $\bA$ is a unit vector with one $1$ in some column and zero elsewhere.
These $r$ unknown coordinates are (very) significant, and all other coordinates are (completely) insignificant.

It turns out that this case is exactly what an \emph{interactive fingerprinting code} can solve. In the interactive fingerprinting code problem \cite{SteinkeU15}, an algorithm $\mathcal{P}$ selects a set of coordinates $\mathcal{S} \subset [n]$ with $|\mathcal{S}| = k$, which is unknown to the fingerprinting code $\mathcal{F}$.\footnote{$\mathcal{P}$ is referred to as the adversary in the original fingerprinting code problem, which would be the opposite for our application. To avoid confusion, we renamed it according to the standpoint here.} Then, the goal of $\mathcal{F}$ is to discover the set $\mathcal{S}$ by making adaptive queries $c^t \in \{\pm 1\}^n$ at each time $t$, and enforcing the requirement that $\mathcal{P}$ must return an answer $a^t$ that is consistent with \emph{some} coordinate in $c^t$, i.e., $a^t = c^j_i$ for some $i\in [n]$. 
Equivalently, $\mathcal{P}$ needs to distinguish between $c^t=(-1,\ldots,-1)$ and $(1,\ldots,1)$.
Importantly, we also impose the constraint that $\mathcal{P}$ can only observe the coordinates $c_i^t$ for $i \in \mathcal{S}$. The attack then proceeds by assigning a score $s_i^t$ to each index $i \in [n]$ at every round $t \in [\ell]$, which corresponds to a measure of the \textit{correlation} between values of the $i$-th index $(c_i^1, \ldots, c_i^t)$ and the responses $(a^1, \ldots, a^t)$ given by $\mathcal{P}$ during the first $t$ rounds. 
It has been shown in~\cite{SteinkeU15} that even under the weak requirement of outputting $-1$ when $c^t=(-1,\ldots,-1)$ and outputting $1$ when $c^t=(1,\ldots,1)$, there is still a nontrivial correlation between the outputs of $\mathcal{P}$ and \emph{some} coordinates in $\mathcal{S}$.
Over time, these correlation scores will accumulate, and can be used by $\mathcal{F}$ to correctly detect coordinates $i \in \mathcal{S}$ with high probability.
It has been shown~\cite{SteinkeU15} that $\mathcal{F}$ can identify the coordinates in $\mathcal{S}$ (with high probability) by making $\O{k^2}$ queries.

In the extreme case where each row of $\bA$ is a unit vector $\be_i$ with a single $1$ in column $i$ and zero everywhere else, we note that the sketch will precisely observe the value of $\bx_i$. 
Furthermore, the task of $\ell_0$ gap norm requires the algorithm to distinguish between the number of non-zeroes $\leq \alpha n$ and $\geq \beta n$, for some constants $0<\alpha<\beta<1$.
This is a stronger requirement than that of $\mathcal{P}$ in the fingerprinting code problem, which merely has to distinguish between all zero queries and all non-zero queries.
Thus, the same attack strategy with the same number of queries applies in this case as well.

\paragraph{Significant coordinates.}
Next, for any sketching matrix $\bA$, let us consider which coordinates of $\bA \bx$ can preserve a significant amount of information about the input $\bx$.
First, if there is a unit vector $\be_i$ (as in the above extreme case), then coordinate $i$ is clearly very significant.
Also, note that since $\bA \bx$ is linear, the query algorithm can recover any $\bw^{\top} \bx$ for $\bw$ in the row span of $\bA$ (i.e., $\exists \by^\top$, s.t., $\bw^{\top}=\by^\top \bA$).
Thus, a relaxation of the unit vector case together with the linearity gives the following definition of ``significance'' of a coordinate $i$:
\[
    \exists \by^\top\in \mathbb{R}^r, (\by^\top \bA)_i^2\geq \frac{1}{s}\cdot \|\by^\top \bA\|_2^2,
\]
for some parameter $s\geq 1$.
That is, there exists a linear combination of the rows such that the $i$-th coordinate is $\ell_2$-heavy.
Equivalently, the \emph{leverage score} of column $i$ is at least $\frac{1}{s}$.
It turns out that if the query vectors were allowed to have \emph{coordinates with real numbers}, this definition captures exactly which coordinates are significant, and is sufficient for proving that if the query vector is supported only on the ``insignificant coordinates'' (in this sense), the query algorithm cannot approximate the $\ell_0$-norm.

However, when $x$ is restricted to having integer coordinates bounded by $\poly(n)$, leverage scores are not sufficient to determine significance of a coordinate.
Consider the matrix $\bA$ with just one row of the form $(C,C,C,\ldots,C,1)$ that has $C$ in every coordinate except that the last coordinate is $1$, for some integer $C\geq 2$.
Every column has a leverage score of only $\O{\frac{1}{n}}$.
On the other hand, when $x$ can only have integer coordinates, $\bA x$ tells us the value of $x_n$ \emph{modulo $C$} (when $C$ is large, this may even completely reveal the coordinate).
This phenomenon can be explained by considering the vector $\bw^\top=\left(1,1,1,\ldots,1,\frac{1}{C}\right)$, which is in the row span of $\bA$.
If we look at the \emph{fractional part} of the inner product $\bw^\top \bx$, the first $n-1$ coordinates never contribute to the value regardless of $\bx$.
In other words, in the \emph{fractional parts} of $\bw^{\top}$, $\left(0,0,0,\ldots,0,\frac{1}{C}\right)$, the last coordinate is in fact very heavy.

This suggests that in general, we should focus on the fractional parts of the vectors in the row span, which motivates us to define the significance of a coordinate $i$ in the following way:
\begin{equation}
\label{eq:frac}
    \exists \by^\top\in\mathbb{R}^r, |\fracpart((\by^\top\bA)_i)|^2 \geq \frac{1}{s}\cdot\|\fracpart(\by^\top\bA)\|_2^2,
\end{equation}
where $\fracpart(\cdot)$ is the fractional part, and when applied on a vector, it is applied coordinate-wise.
It turns out that this definition captures our needs, and is what we will use for our main result over the integers.

\paragraph{Matrix pre-processing.}
To facilitate the analysis of the attack, we will first ``pre-process'' the sketching matrix $\bA$ to obtain a new matrix $\bA'$ that separates the significant coordinates and the insignificant coordinates, while not weakening the sketch $\bA \bx$.

Let us consider the following pre-processing procedure on $\bA$: while there exists a column $i \in [n]$ such that $\eqref{eq:frac}$ holds, we zero out the $i$-th column of $\bA$ and add $i$ to the set of significant coordinates $\mathcal{S}$. 
Note that new columns may become significant when we zero out a column, and the procedure is applied iteratively on the remaining matrix until no column satisfies $\eqref{eq:frac}$.
Finally, for each coordinate $i\in \mathcal{S}$, we add a new row $\be_i$.
Thus, the overall pre-processing can be viewed as follows: we find the significant coordinates; since $\bA \bx$ may preserve a significant amount of their information, we might as well just strengthen the sketch so that it actually stores them precisely; then the rest of the sketch is made independent of them by zeroing out the corresponding columns.

Let $\bA'$ denote the matrix after these operations. 
Without loss of generality, we can assume that the actual sketching matrix is $\bA'$ instead of $\bA$, since $\bA'\bx$ can recover $\bA \bx$ (as we can just add the new rows $\be_i$, with the correct weights, back to each row where column $i$ was zeroed out), it makes the algorithm at least as powerful as it was. 
The new sketching matrix $\bA'$ has the following form:
\[
\bA' = \begin{bmatrix}
    \bD \\
    \mathbf{S}
\end{bmatrix},
\]
where we note that no column is significant in the sense of $\eqref{eq:frac}$ for $\bD$, and $\mathbf{S}$ has at most one non-zero entry $1$ in each row and column. Moreover, the non-zero columns of $\bD$ and $\bS$ are disjoint. We refer to $\bD$ as the dense part and $\mathbf{S}$ as the sparse part, and note that the set of significant coordinates $\mathcal{S}$ is precisely the set of non-zero columns in the sparse part $\bS$.

Note that the sparse part is exactly the extreme case that we discussed earlier, and $\mathcal{S}$ can be learned using the fingerprinting code \emph{if there were no dense part}. 
Moreover, we show that the definition of significant coordinates and the pre-processing procedure guarantee that the sparse part is small, $|\mathcal{S}|\ll n$, so that after learning $\mathcal{S}$ and zeroing out these coordinates in the query, we will not be left with a trivial problem. Roughly speaking, this is shown by proving that under the uniform distribution of  $\bx \in \{-1, 0, 1\}^n$, if a column $i$ satisfies~\eqref{eq:frac}, then $\bA \bx$ must have a nontrivial mutual information with $\bx_i$, $I(\bA \bx; \bx_i) \geq \Omega\left(\frac{1}{s}\right)$. 
Then, if the pre-processing algorithm removes $T$ columns iteratively, by applying the chain rule for mutual information, we can argue that the mutual information between $\bA \bx$ and all these $T$ corresponding coordinates is at least $\Omega\left(\frac{T}{s}\right)$.
On the other hand, it can be at most $\O{r\log n}$, as $\bA\bx$ can be encoded in $\O{r\log n}$ bits. 
Hence, we can add at most $T=\O{rs\log n}$ rows to the sparse part.

\paragraph{Description of the attack.}
Lastly, it remains to show that the dense part (insignificant coordinates) cannot be useful to the algorithm, by carefully selecting the query distribution.
More specifically, we will design a family of distributions $\mathcal{D}$ over $\{-R,-(R-1),\ldots,R\}$ for some integer $R=\poly(n)$ bounded by a small polynomial in $n$, with the following properties:
     \begin{enumerate}
         \item For $D_p \in \mathcal{D}$ where $p \in [\alpha, \beta]$ for some constant $0<\alpha<\beta<1$, we have $\PPPr{X \sim D_p}{X = 0} = p$;
         \item For any $p, q \in [\alpha, \beta]$, we have $d_{\mathrm{tv}}(\bD \bx_p, \bD \bx_q) \le \frac{1}{\poly(n)}$ for $\bx_p \sim D_p^n$ and $\bx_q \sim D_q^n$.
     \end{enumerate}

We will give more details about how to construct such a family of distributions later in this section.

Given such a family, we consider the sketch $\bA' \bx = \begin{bmatrix}
    \bD \bx \\
    \mathbf{S} \bx \\
\end{bmatrix}$ for query vectors $\bx \sim D_p^n$ for different $p$. 
From the property of the distribution family $\mathcal{D}$, we know that the (marginal) distribution of $\bD\bx$ is almost identical regardless of the value of $p$. 
Moreover, since $\bD$ and $\bS$ have disjoint nonzero columns, $\bD\bx$ and $\bS\bx$ are independent conditioned on $p$.
This allows us to conclude that if we sample the queries from these distributions, then the algorithm must approximate $\|\bx\|_0$ by only looking at the sparse part $\mathbf{S} \bx$.
It turns out that these distributions $D_p$ can be ``integrated'' into the fingerprinting code, so that the dense part cannot help the algorithm when we attack the sparse part.
This allows us to gradually identify $\mathcal{S}$, and eventually zero out all of these significant coordinates.
Finally, when we make one more query with all coordinates in $\mathcal{S}$ zeroed out, the algorithm will have to produce an output based only on $\bD \bx$, and will thus fail with high probability.

\paragraph{Constructing hard distributions for the insignificant coordinates.}
%As we discussed earlier, for a matrix $\bA$ that satisfies the condition~\eqref{eq:frac} for a specific parameter $s$, we design a family of distributions $\mathcal{D} = \{D_p\}$ which satisfies the following properties:
%\begin{enumerate}
%    \item For every $p \in [\alpha, \beta]$,  we have that $\PPr{X = 0} = p$ for $X \sim D_p$.
%    \item For every $p \in [\alpha, \beta]$, , $D_p$ is a symmetric distribution supported on $\{-R, -R + 1, ..., R\}$, for $R = \O{r \log n}$.

%    \item For every pair of $p, q \in [\alpha, \beta]$, the total variation distance of $\bD \bx_1$ and $\bD\bx_2$ is at most $1 / \poly(n)$, where $\bx_1 \sim D_{p}^n$ and $\bx_2 \sim D_{q}^n$.
%\end{enumerate}

We wish to construct a family of distributions such that the total variation distance between $\bD \bx_p$ and $\bD \bx_q$ for $\bx_p\sim D_p^n,\bx_q\sim D_q^n$ is small.
We will rely on the following property of $\mathcal{D}$: for every pair $D_p, D_q \in \mathcal{D}$, we have that 
\[
\EEx{X\sim D_p}{X^k}=\EEx{X\sim D_q}{X^k}
\]
for all $k\in[K]$, i.e. the first $K = \O{r \log n}$ moments of $D_p$ and $D_q$ match. 
In fact, we will make all distributions $D_p\in \mathcal{D}$ \emph{symmetric}, i.e., $D_p(t)=D_p(-t)$.
Thus, all odd moments are zero, and hence, equal.
Then for the even moments, the condition is equivalent to  $\sum_{i=0}^R i^k\cdot(D_p(i)-D_q(i))=0$ for $k \le K$. Our construction is based on the following fact (e.g., see Claim 1 in~\cite{LarsenWY20}): there exists a polynomial $Q$ with degree at most $R-\Omega(\sqrt{R})$ such that
\[
|Q(0)|=\Omega(1)\qquad\mathrm{and}\qquad\sum_{i=0}^R\left\lvert\binom{R}{i}\cdot Q(i)\right\rvert=\O{1}.
\]
The degree bound on $Q$ further implies that
\[
\sum_{i=0}^R(-1)^i\binom{R}{i}\cdot Q(i)\cdot i^t=0
\]
for all non-negative integers $t < R-\mathrm{deg}(Q)$, since $Q(i)\cdot i^t$ is a polynomial of degree strictly less than $R$, and $\sum_{i=0}^R (-1)^i \binom{R}{i}\cdot P(i)=0$ holds for any polynomial $P$ of degree $<R$.

Hence, we will set $R=\Theta(K^2)$ for a sufficiently large leading constant, and define the distribution family $\mathcal{D} = \{D_p\}$ such that $D_p(i)=D(i)+c_p\cdot (-1)^i\binom{R}{i}\cdot Q(i)$, for some distribution $D$ and constants $c_p$.
The difference between the probabilities $D_p(i)$ and $D_q(i)$ is precisely $c_p-c_q$ times $(-1)^i\binom{R}{i}\cdot Q(i)$.
Then we can ensure that our moment matching condition is satisfied, since 
\[
    \sum_{i=0}^R i^k\cdot (D_p(X)-D_q(X))=(c_p-c_q)\cdot \sum_{i=0}^R i^k\cdot (-1)^i\binom{R}{i}\cdot Q(i)=0,
\]
for $k\leq K\leq\O{\sqrt{R}}$.
Furthermore, the bounds on $\sum_{i=0}^R \left|\binom{R}{i}\cdot Q(i)\right|$ and $|Q(0)|$ ensure that the range of the distribution family $\beta - \alpha$ can be made $\Omega(1)$ by carefully picking the base distribution $D$ (recall that $\alpha,\beta$ are the smallest and the largest probabilities at $0$ over all distributions in the family).

\paragraph{Bounding the total variation distance.}
Let $P = D_p$ and $Q = D_q$ be distributions from family $\mathcal{D}$ that match the first $K$ moments for some $p, q \in [\alpha, \beta]$. Suppose $P^n$ and $Q^n$ are probability distributions of $n$-dimensional vectors, where each entry is drawn independently from $P$ and $Q$, respectively. As before, let $\bD$ denote the dense matrix such that no column satisfies~\eqref{eq:frac} with parameter $s$. For $\bx \sim P^n$ and $\bx' \sim Q^n$, let $P_{\bD}$ and $Q_{\bD}$ be the probability distributions of $\bD \bx$ and $\bD \bx'$. Now, we will argue that $d_{\mathrm{tv}}(P_{\bD}, Q_{\bD}) \leq \frac{1}{\poly(n)}$. To see this, we use the following observation from Fourier analysis:
\begin{align*}
|P_{\bD}(x) - Q_{\bD}(x)| &= \left| \frac{1}{(2\pi)^r} \int_{[-\pi, \pi)^r} e^{i \langle \bu, \bx \rangle} \left(\widehat {P_{\bD}}(\bu) - \widehat {Q_{\bD}}(\bu) \right) d\bu \right| \\
&\leq \frac{1}{(2\pi)^r} \int_{[-\pi, \pi)^r} \left |\widehat {P_{\bD}}(\bu) - \widehat {Q_{\bD}}(\bu) \right | d \bu
\end{align*}
where the last inequality follows by triangle inequality. So, to bound the difference of $P_{\bD}(\bx)$ and $Q_{\bD}(\bx)$ for a particular value $\bx$, we just need to upper bound the quantity $\left |\widehat {P_{\bD}} (\bu) - \widehat {Q_{\bD}}(\bu) \right |$. Let $P_i = \PPr{X = i}$ and $\fracpart_{2\pi}(x)=2\pi\cdot\fracpart\left(\frac{x}{2\pi}\right)\in[-\pi,\pi)$. Then, we can then express $\widehat {P_{\bD}}(\bu)$ (and similarly for $\widehat {Q_{\bD}}(\bu)$) as follows: 
\begin{align*}\widehat {P_{\bD}}(\bu) = \mathbb{E}_{\bz \sim P_{\bD}}\left[ e^{-i\langle \bu, \bz \rangle}\right] &= \mathbb{E}_{\bx \sim P^n} \left[e^{-i\langle \bu, \bD \bx \rangle} \right] \\ &= \prod_{j \in [n]} \sum_{k \geq 0} P_k \cdot \cos \left( k \cdot \langle \bu, \bD^{(j)} \rangle  \right) \\
&=\prod_{j\in[n]}\sum_{k\ge 0}P_k\cdot\cos\left(k\cdot\fracpart_{2\pi}(\langle\bu,\bD^{(j)}\rangle)\right).
\end{align*}
 where the second equality follows since our chosen distribution $D_p$ is symmetric and we draw each coordinate $\bx_j \sim D_p$ independently. Now, by the Taylor expansion $\cos(x) = 1 - \frac{x^2}{2!} + \frac{x^4}{4!} - \frac{x^6}{6!} + \ldots$, we can write 
\begin{align*}
\widehat{P_{\bD}}(\bu)=\prod_{j\in[n]}\sum_{k\ge 0}\left(\sum_{m\ge 0}P_m\cdot m^{2k}\right)\cdot\frac{\left(\fracpart_{2\pi}(\langle\bu,\bD^{(j)}\rangle)\right)^{2k}}{(2k)!}\cdot(-1)^k,\\
\widehat{Q_{\bD}}(\bu)=\prod_{j\in[n]}\sum_{k\ge 0}\left(\sum_{m\ge 0}Q_m\cdot m^{2k}\right)\cdot\frac{\left(\fracpart_{2\pi}(\langle\bu,\bD^{(j)}\rangle)\right)^{2k}}{(2k)!}\cdot(-1)^k
\end{align*}
Let $M_P(2k)= \left(\sum_{m\ge 0}P_m\cdot m^{2k}\right)$ and $M_Q(2k) = \left(\sum_{m\ge 0}Q_m\cdot m^{2k}\right)$ denote the $2k$-th moment of $P$ and $Q$, respectively. At this point, our proof makes use of two key properties to upper bound $\left |\widehat {P_{\bD}}(\bu) - \widehat {Q_{\bD}}(\bu) \right |$:

\begin{enumerate}
    \item \textbf{Bounded fractional parts.} First, we recall that $\bD$ satisfies $|\fracpart(\by^\top\bD)_j|^2 \le \frac{1}{s}\cdot\|\fracpart(\by^\top\bD)\|_2^2$ for all $y \in \mathbb{R}^r$ and $j \in [n]$. Then, if there exists some index $j \in [n]$ such that $|\fracpart_{2\pi} (\langle\bu,\bD^{(j)}\rangle)| \geq \frac{1}{K}$ (for some chosen threshold $t$), we can use the above property of $\left\| \fracpart\left(\left\langle\frac{\bu}{2\pi},\bD\right\rangle\right)\right\|_2^2$ to upper bound $\widehat {P_{\bD}}(\bu)$ and $\widehat {Q_{\bD}}(\bu)$.

    \item \textbf{Moment matching.} Alternatively, suppose there is no such index $j$; then we can use the fact that $M_{P}(2k) = M_Q(2k)$ for $k \leq K/2$, so we have that the first $K/2$ terms of
   
   \begin{align*}
   \sum_{k\ge 0}M_P(2k)\cdot\frac{\left(\fracpart_{2\pi}(\langle\bu,\bD^{(j)}\rangle)\right)^{2k}}{(2k)!}\cdot(-1)^k \\
   \sum_{k\ge 0}M_Q(2k)\cdot\frac{\left(\fracpart_{2\pi}(\langle\bu,\bD^{(j)}\rangle)\right)^{2k}}{(2k)!}\cdot(-1)^k
    \end{align*}
    % $\frac{\left(\fracpart_{2\pi}(\langle\bu,\bD^{(j)}\rangle)\right)^{2k}}{(2k)!}\cdot(-1)^k$  

    are exactly the same. By combining this fact with our assumption that $|\fracpart_{2\pi} (\langle\bu,\bD^{(j)}\rangle)| < \frac{1}{K}$ for every $j \in [n]$, we obtain the desired upper bound for this case as well.
    
\end{enumerate}
For the full argument, we refer the readers to Section~\ref{sec:distribution}.
Finally, since $\bD$ is a matrix in $\mathbb{Z}^{r \times n}$ with entries bounded in $\poly(n)$, we know that the total support size of $P_{\bD}$ and $Q_{\bD}$ is $n^{\O{r}}$. So, after we compute an upper bound for $\left |\widehat P_{\bD}(\bu) - \widehat Q_{\bD}(\bu) \right |$, we can finish the argument by simply union-bounding over the total size of the support of $\bD \bx$ to obtain the upper bound of $d_{\mathrm{tv}}(P_{\bD}(\bx), Q_{\bD}(\bx)) \leq \frac{1}{\poly(n)}$ for some choice of parameters $K$ and $s$.

\subsection{Overview of Attack over Finite Fields}
When the sketching matrix $\bA \in \mathbb{F}_p^{r \times n}$ for a fixed prime $p$, our attack is based on the following crucial observation: suppose that $U$ and $R$ are the two subsets of indices of columns of $\bA$ such that $\bA^U$ and $\bA^R$ have the same column span. Then, if $\bx \sim \mathbb{F}_p^{|U|}$ and $\bx' \sim \mathbb{F}_p^{|R|}$ are sampled uniformly at random, we can show that $\bA^U \bx$ and $\bA^R \bx'$ are identically distributed. With this in mind, note that if we can find an independent set of columns $T$ with $|T| = r$, then the streaming algorithm $\mathcal{A}$ will not be able to distinguish $\bA^T \bx'$ where $\bx' \sim \mathbb{F}_p^{r}$ and $\bA \bx$ where $\bx \sim \mathbb{F}_p^{n}$. 
% Note that without loss of generality we can assume $n = Cr$ for some large constant $C$ (as we can always zero the last $n - Cr$ coordinates of $\bx$). 
Hence, $\mathcal{A}$ must fail on one of the input distributions (we assume $n \ge 2r$).
% as with high constant the $\ell_0$ value of these two distributions has a constant gap.
Therefore, our goal now is to find such a column-independent set. %  

The way we search for this column independent set is as follows: suppose that the set $T$ is what we have maintained up to now. Then let $R$ be a random sample of $2r$ columns outside $T$ and $R^i$ is the first $i$ column of $R$. Let $\mu_i$ denote the distribution of $f(\bA \bx^{(i)})$, where $\bx^{(i)} \in \mathbb{F}_p^n$ is the random vector that on the support of $T \cup R^i$. 
From the correctness guarantee we must have the total variation distance $d_{\mathrm{tv}}(\mu_0, \mu_{2r - 1}) = \Omega(1)$ (otherwise we find a distribution that $\mathcal{A}$ fails with constant probability immediately). Then from triangle inequality we have 
% \[
% \sum_i |\mu_{i} - \mu_{i + 1}| \ge  |\sum_i (\mu_{i} - \mu_{i + 1})| \ge \Omega(1).
% \]
\[
\sum_i d_{\mathrm{tv}}(\mu_i, \mu_{i + 1}) \ge  d_{\mathrm{tv}}(\mu_0, \mu_r) = \Omega(1).
\]
From this we get there must exist one $j$ such that $ d_{\mathrm{tv}}(\mu_{j - 1}, \mu_{j}) \ge \Omega\left(\frac{1}{r}\right)$, which means that the $j$-th column in $R$ must be linear independent in $T$. Note that from the result in statistical testing we can distinguish this case using $\widetilde{O}(r^2)$ samples with error probability at most $1/\poly(r)$. Hence, we can enumerate the index $i$ and do the testing between $\mu_i$ and $\mu_{i + 1}$ to find such column index $j$.

  The above procedure requires $\tO{r^4}$ total number of queries, as we need to find $r$ columns and in each step, we make $2r \cdot \tO{r^2} = \tO{r^3}$ queries. However, the dependence of $r$ can be further improved. Note that in the worst case $\max_i \{d_{\mathrm{tv}}(\mu_{i - 1}, \mu_{i})\} = \Theta\left(\frac{1}{r}\right)$, we can randomly sample $\O{1}$ indices to find such index $j$, which suggests a better dependence of $r$. Indeed, we show that there must exist $\ell$ for which there exist at least $2^{\ell-1}$ indices $i$ such that $d_{\mathrm{tv}}(\mu_i,\mu_{i+1})\in\left[\frac{1}{2^{\ell+3}\log r},\frac{1}{2^{\ell+2}\log r}\right)$. Hence, we can make a guess of such $\ell$ and note that for each different guess, since the range of the total variation distance changes, we can use a different number of the samples in the testing procedure, which results in an overall $\tO{r^3}$ number of queries.

\section{Preliminaries}
For a positive integer $n>0$, we write $[n]$ to denote the set $\{1,2,\ldots,n\}$. 
We use the notation $\poly(n)$ to denote a fixed polynomial in $n$ and $\polylog(n)$ to represent $\poly(\log n)$. 
We say an event $\calE$ occurs with high probability if $\PPr{\calE}\ge 1-\frac{1}{\poly(n)}$, when the dependent variable $n$ is clear from context. 

\subsection{Interactive Fingerprinting Codes}
\label{sec:prelim}

An interactive fingerprinting code $\mathcal{F}$ is an efficient adaptive algorithm that defeats any adversary $\mathcal{P}$ in the following two-player game. The adversary $\mathcal{P}$ first selects a secret subset of indices $\mathcal{S} \subset [N]$, where $|\mathcal{S}| = n$. Then, the goal of $\mathcal{F}$ is to construct an adaptive sequence of queries $\{c^t\}_{t \in [\ell]}$ to learn (or ``accuse'') all of the indices $i \in \mathcal{S}$, while making few false accusations (i.e., incorrectly accusing some $i \not \in \mathcal{S}$) in the process. Specifically, in each round $t \in [\ell]$, the interactive fingerprinting code $\mathcal{F}$ selects a query vector $c^t \in \{\pm 1\}^N$, and the adversary $\mathcal{P}$ observes only the coordinates $c_i^t$ for those $i \in \mathcal{S}$, and has no knowledge of $c_i^t$ for $i \not \in \mathcal{S}$. Then, the adversary must respond with an answer $a^t$ that is \textit{consistent} with some coordinate of $c^t$ such that $a^t = c^t_i$ for some $i \in \mathcal{S}$. More concretely, if all of the coordinates of $c^t = 1^N$ then $a^t$ must be $1$, or if $c^{t} = (-1)^N$, then $a^t$ must return  $-1$.  

Informally, the interactive fingerprinting attack of \cite{SteinkeU15} proceeds by assigning a score $s_i^t$ to each index $i \in [N]$ at every round $t \in [\ell]$, which corresponds to a measure of the \textit{correlation} between values of the $i^{\textrm{th}}$ index $(c_i^1,..., c_i^t)$ and the responses $(a^1,..., a^t)$ given by the adversary during the first $t$ rounds. The interactive fingerprinting code $\mathcal{F}$ accuses coordinates $i \in [N]$ whose score $s_i^t$ exceeds a threshold $\sigma$ at some point during the sequence of queries. Using this approach, combined with an appropriate hard distribution for inputs $c^t \in \{\pm 1\}^N$, \cite{SteinkeU15} shows that for every $N \in \mathbb{N}$, there exists an interactive fingerprinting code that makes $\ell = \O{n^2\log\frac{1}{\delta}}$ queries and, except with negligible probability, identifies all of $\mathcal{S}$ and makes at most $\frac{N\delta}{1000}$ false accusations. Moreover, their attack satisfies a \textit{robustness} property: the result above holds even when the fingerprinting adversary $\mathcal{P}$ only provides a response $a^t$ which is consistent with some coordinate $c_i^t$ in at least $(1-\beta)\ell$ of the rounds, for any $\beta < 1/2$.

We provide a brief overview of the  game, as well as the attack of \cite{SteinkeU15} here, as a reference.

\begin{definition}[Interactive Fingerprinting Code Game] The Interactive Fingerprinting Code problem is defined via the following game.
\begin{enumerate}
    \item First, the adversary $\mathcal{P}$ selects a subset of users $S^1 \subseteq [N]$, $|S^1| = n$, which is unknown to the fingerprinting code $\mathcal{F}$. 
    \item In each round $j = 1,..., \ell$:
    \begin{itemize}
        \item $\mathcal{F}$ outputs a column vector $c^j \in \{\pm 1\}^N$.
        \item Let $c_{S^j}^j \in \{\pm 1\}^{|S^j|}$ be the restriction of $c^j$ to coordinates in $S^j$: only this restricted copy $c_{S^j}^j$ is given to the adversary $\mathcal{P}$ in each round.
        \item Then, $\mathcal{P}$ outputs $a^j \in \{\pm 1\}$, which is observed by $\mathcal{F}$.
        \item Finally, $\mathcal{F}$ accuses a set of users $I^j \subseteq [N]$, and sets $S^{j+1} = S^j \setminus I^j$ as the current ``undiscovered'' set of coordinates/users.
    \end{itemize}
\end{enumerate}
\end{definition}

\paragraph{Construction of attack, c.f., \cite{SteinkeU15}}
For $0 \leq a < b \leq 1$, let $P_{a,b}$ be the distribution with support $(a,b)$ and probability density function $\mu(p) = \frac{C_{a,b}}{\sqrt{p(1-p)}}$. For $\alpha, \zeta \in (0,\frac{1}{2})$, let $\overline{P_{\alpha, \zeta}}$ be the distribution on $[0,1]$ such that it returns a sample from $D_{\alpha, 1-\alpha}$ with probability $1-2\zeta$, and returns $0/1$ each with probability $\zeta$. Furthermore, let $\phi:\{-1, 1\} \rightarrow \mathbb{R}$ be defined by $\phi^{0}(c) = \phi^{1}(c) = 0$ and for $p \in (0,1)$, we have $\phi^{p}(1) = \sqrt{\frac{1-p}{p}}$ and $\phi^{p}(0) = -\sqrt{\frac{p}{1-p}}$. We consider the following parameter regime: 

$$\alpha = \frac{(\frac{1}{2} - \beta)}{n}$$
$$\zeta = \frac{3}{8} + \frac{\beta}{4}$$
$$\sigma = \O{\frac{n}{(\frac{1}{2}-\beta)^2} \log \left(\frac{1}{\delta}\right)}$$
$$\ell = \O{\frac{n^2}{(\frac{1}{2} - \beta)^4} \log \left( \frac{1}{\delta} \right)}$$

\begin{enumerate}
    \item Let $s_i^0 = 0$ for every $i \in [N]$.
    \item For $j = 1,..., \ell$:
    \begin{itemize}
        \item Draw $p^j \sim \overline{P_{\alpha, \zeta}}$ and $c_{1\cdots N}^j \sim p^j$.
        \item Issue $c^j \in \{\pm 1\}^N$ as a challenge and get response $a^j \in \{\pm 1\}$.
        \item For every $i \in [N]$, update score $s_i^{j} = s_i^{j-1} + a^j \cdot \phi^{p^j}(c_i^j)$
    \end{itemize}
\end{enumerate}

Importantly, the attack enforces the following completeness and soundness properties:
\begin{itemize}
    \item \textbf{Completeness:} If $i \in S^1$, the score of user $i$ will exceed some chosen threshold at \textit{some} step $j \in [\ell]$, i.e. with high probability, there exists $j$ such that $s_i^j > \sigma$.
 
    \item \textbf{Soundness:} Alternatively, if $i \not \in S^1$, the score $s_i^j$ will not exceed $\sigma$ with high probability. The argument uses the fact that the responses of $\mathcal{P}$ cannot have high correlation with $(c_i^1, ..., c_i^{\ell})$ if $\mathcal{P}$ never sees this information. 
\end{itemize}
    
% \end{definition}

\subsection{Preliminaries from Information Theory}
We recall the following preliminaries from information theory. 
\begin{definition}[Entropy and conditional entropy]
The \emph{entropy} of a random variable $X$ taking on possible values in a finite space $\Omega$ is defined as 
\[H(X):=\sum_{x\in\Omega} p(x)\log\frac{1}{p(x)},\]
where $p(x)=\PPr{X=x}$ is the probability mass function of $X$.  
The \emph{conditional entropy} of $X$ with respect to a random variable $Y$ is defined as 
\[H(X|Y)=\mathbb{E}_{y}{H(X|Y=y)},\]
where $H(X|Y=y):=\sum_{x\in\Omega} p(x|y)\log\frac{1}{p(x|y)}$, for the conditional probability mass function $p(x|y)$. 
\end{definition}

\begin{definition}[Mutual information and conditional mutual information]
We define the \emph{mutual information} between random variables $X$ and $Y$ by
\[I(X;Y)=H(X)-H(X|Y)=H(Y)-H(Y|X)=I(Y;X).\]
We define the \emph{conditional mutual information} between $X$ and $Y$ conditioned on a random variable $Z$ by 
\[I(X;Y|Z)=H(X|Z)-H(X|Y,Z).\]
\end{definition}
%\samson{Data processing inequality} 
\begin{theorem}[Data-processing inequality]
    Let $X, Y, Z$ be random variables such that $X \rightarrow Y \rightarrow Z$ forms a Markov Chain, i.e., $X$ and $Z$ are conditionally independent given $Y$. Then, we have $I(X; Y) \geq I(X;Z)$.
\end{theorem}

\begin{theorem}[Chain rule for mutual information]
\thmlab{thm:chain:information}
For random variables $X_1,\ldots,X_n,Z$, we have
\[I(X_1,\ldots,X_n;Z)= \sum_{i = 1}^n I(X_i; Z | X_1, ..., X_{i-1}).\]
\end{theorem}

\section{Pre-processing the Sketching Matrix}
\label{sec:pre-processing}
% \begin{definition}
%     Given a matrix $A \in \mathbb{Z}^{r \times n}$, we say a row $A_i$ is \textit{sparse} if it has at most $s$ non-zero entries. Otherwise, if row $A_i$ has more than $s$ non-zero entries, we call that row \textit{dense}.
% \end{definition}
Our attack will rely on pre-processing and decomposing the sketching matrix $\bA$ into \textit{sparse} part and a \textit{dense} part, which will consist of disjoint sets of non-zero indices. This pre-processing procedure will have the property that it can only make the streaming algorithm stronger, by potentially allowing the algorithm to observe more entries of the input vector $\bx^{(t)}$. More formally, our new matrix $\bA'$ will satisfy several key properties, as stated in the next lemma.
\begin{lemma}
\lemlab{lem:pre}
For any algorithm $\mathcal{A}$ with sketching matrix $\bA \in \mathbb{Z}^{r \times n}$, there is a pre-processing procedure that produces a new matrix $\bA' \in \mathbb{Z}^{r' \times n}$ for $r' = \O{rs\log n}$ satisfying the following properties:
\begin{enumerate}
    \item The $\bA'$ has the form $\begin{bmatrix}
        \bD \\
        \bS
    \end{bmatrix}$ where the $\bD$ and $\bS$ are column-disjoint.
    \item We have $|\fracpart(\by^\top\bD)_j|^2 \le \frac{1}{s}\cdot\|\fracpart(\by^\top\bD)\|_2^2$ for all $\by\in\mathbb{R}^{r}$ and $j\in[n]$ 
    % where $\fracpart(x) = x-\mathsf{int}(x) \in (-\frac{1}{2}, \frac{1}{2}]$ where $\mathsf{int}(x)$ is the closest integer number to $x$
    \item Each row and column of $\bS$ has at most one non-zero entry.
\end{enumerate} 
 Moreover, without loss of generality, we can assume the algorithm $\mathcal{A}$ uses sketching matrix $\bA'$ instead of $\bA$.
\end{lemma}

\begin{proof}
    We consider the following procedure: we start with the original sketching matrix $\bA$, and for each time $t$, we identify a columns $j_t\in[n]$ such that $|\fracpart(\by^\top\bD^{(t-1)})_{j_t}|^2 > \frac{1}{s}\cdot\|\ \fracpart(\by^\top\bD^{(t-1)})\|_2^2$ where $\bD^{t - 1}$ is the first $r$ rows of $\bA^{(t - 1)}$, we zero the $j_t$-th column of $\bA^{t - 1}$ and add a new row $\be_{j_t}$ to the matrix $\bA^{(t - 1)}$. We then denote the new resulting matrix as $\bA^{(t)}$. Suppose that the above procedure ends in the iteration $T$. From 
    \lemref{lem:remove:s:heavy:frac} we have $T = \O{rs \log n}$.
    
    Let $\bD = \bD^{(T)}$ and $\bS$ be the remaining rows of $\bA^{(T)}$. Then we have    $\bA' = \begin{bmatrix}
        \bD \\
        \bS
    \end{bmatrix}$ has at most $r + T = \O{rs \log n}$ columns. Then from the procedure it is easy to see that the $\bD$ and $\bS$ are column-disjoint and each row and column of $\bS$ has at most one non-zero entry.

    At this point, it remains for us to show why we can assume that the sketching matrix used by $\mathcal{A}$ is $\bA'$ instead of $\bA$. Suppose that the algorithm $\mathcal{A}$ uses the sketching matrix $\bA$ and estimator $f$ on $\bA \bx$, then we consider the following equivalent form of the algorithm $\mathcal{A}$, where it uses the sketching matrix $\bA'$ and another estimator $g$. Given an $\bA'\bx$, estimator $g$ will first invert the row operations that transfer $\bA' \bx$ to $\bA \bx$ (recall that all of the operations we apply on $\bA$ are invertible) and output the value $f(\bA \bx)$. From the definition of $g$, we immediately get that $g(\bA' \bx) = f(\bA \bx)$ for every input $\bx$, which means we can assume $\mathcal{A}$ has the form $g(\bA' \bx)$ without loss of generality.
\end{proof}

\subsection{Bounding the Number of Added Rows}
Let $\fracpart(x) = x-\mathsf{int}(x) \in (-\frac{1}{2}, \frac{1}{2}]$ where $\mathsf{int}(x)$ is the closest integer number to $x$ and for a vector $\bx\in\mathbb{R}^n$, let $\fracpart(\bx)\in\mathbb{R}^n$ be the coordinate-wise fractional parts of $\bx$, i.e., $\fracpart(\bx)_j=\fracpart(x_j)$. 
\begin{lemma}
\lemlab{lem:random:vec:info:frac}
Let $\bA\in \mathbb{Z}^{r\times n}$ be a fixed matrix and let $\bx\in\{-1, 0, 1\}^n$ such that each coordinate is chosen independently and with probability $1 - \frac{2c}{s}$, $x_i = 0$, and with probability $\frac{2c}{s }$, $x_i = 1$ or $-1$ with equal probability, where $c$ is a sufficiently small constant. 
Suppose there exists $\by\in\mathbb{R}^r$ and $j\in[n]$ such that for $|\fracpart((\by^\top\bA)_j)|^2\ge\frac{1}{s}\cdot\|\fracpart(\by^\top\bA)\|_2^2$. 
Then we have $I(\bA\bx; x_j)=\Omega\left(\frac{1}{s}\right)$. 
\end{lemma}
\begin{proof}
Observe that since by the data-processing inequality, 
\[I(\bA\bx; x_j)\ge I(\by^\top\bA\bx; x_j)\ge I(\fracpart(\by^\top\bA\bx); x_j),\]
then it suffices to show $I(\fracpart(\by^\top\bA\bx); x_j)=\Omega\left(\frac{1}{s}\right)$. 
Let $\ba=\by^\top\bA\in\mathbb{R}^n$.
We sample column vectors $\bx^{(1)},\ldots,\bx^{(t)}$, such that all coordinates are selected randomly from $\{-1, 0, 1\}$ from the distribution in the lemma statement but the $j$-th coordinate is the same in all $t$ vectors. 
Since the marginal distributions are the same for each $\bx^{(k)}$, we have that $I(\fracpart(\langle\ba,\bx\rangle); x_j) = I(\fracpart(\langle\ba,\bx^{(k)})\rangle; x_j)$. 

We next claim that, by looking at $t = \O{\log s}$ samples $\fracpart(\langle\ba,\bx^{(k)})\rangle$ from $k = 1, 2, \cdots t$, we can determine the value of $x_j$ with probability at least $1 - 1/\poly(s)$, which means that

% By looking at the average $\frac{1}{t}\sum_{k=1}^t\fracpart(a_ix_i^{(k)})$ and setting $t=10^{20}s^2\log^{3/2}n$, we claim that 
\[I(\fracpart(\langle\ba,\bx^{(1)}\rangle_,\ldots,\fracpart(\langle\ba,\bx^{(t)}\rangle); x_j)=\Omega(\log s /s).\]

\paragraph{Mutual information from independent instances.}

Firstly, note that if $\|\fracpart(\ba)\|_2^2>s$, then $\frac{1}{s}\cdot\|\fracpart(\ba)\|_2^2>1$ and so there cannot exist $y$ such that $|\fracpart(\ba_j)|^2\ge\frac{1}{s}\cdot\|\fracpart(\ba)\|_2^2$. 
Thus it suffices to consider $\|\fracpart(\ba)\|_2^2\le s$.

We next consider one of the instances $\bx$, let $S$ be the set of indices such that $x_i \ne 0$ and $i \ne j$. Then we have $\Ex{\|\fracpart(\ba_S)\|_2^2}\le \frac{2c}{s} \|\fracpart(\ba)\|_2^2$, by Markov's inequality we have with probability at least $0.99$, $\|\fracpart(\ba_S)\|_2^2 \le \frac{1}{100s} \|\fracpart(\ba)\|_2^2$. Condition on this event, by Markov's inequality again we can have with probability at least $0.9$,
\[
\sum_{i \in S} \left( \fracpart(a_i) \cdot x_i \right)^2 \le \frac{1}{10s} \|\fracpart(\ba)\|_2^2 \le \frac{1}{10}.
\]
which means that $\fracpart(\langle\ba,\bx\rangle) = \fracpart(\fracpart(a_j)\cdot x_j + \alpha)$ where $\alpha = \sum_{i \in S} \fracpart(a_i) \cdot x_i \le \frac{1}{3\sqrt{s}}\|\fracpart(\ba)\|_2 \le \frac{1}{3} \left| \fracpart(a_j)\right|$.

% Consider the case where $\fracpart(-a_j)$ also have value $\Omega(\fracpart(a_j))$. In this case, we can determine whether $x_j = 0$ or not with high constant probability by looking at the value $\fracpart(\langle\ba,\bx\rangle)$.
% For the other case, we can still determine the case when $x_j \ne 0$ and $x_j$ has the same sign as $a_j$ with high constant probability by looking at the value of $\fracpart(\langle\ba,\bx\rangle)$. 

Condition on the above events happen, consider the case where $x_j = 0$. In this case, we have $|\fracpart(\langle\ba,\bx\rangle)| \le \frac{1}{3} \left| \fracpart(a_j)\right|$. On the other hand, if $x_j \ne 0$, we will have $|\fracpart(\langle\ba,\bx\rangle)| \ge \frac{2}{3}  \left| \fracpart(a_j)\right|$ and the sign of $|\fracpart(\langle\ba,\bx\rangle)|$ is the same as $x_j$, which means that we can determine the value of $x_j$ by looking at the value of $\fracpart(\langle\ba,\bx\rangle)$.

The above procedure succeeds with high constant probability, to boost the success probability, we can instead look at the majority of the outputs by $\O{\log s}$ independent instances 
\[
(\fracpart(\langle\ba,\bx^{(1)}\rangle_,\ldots,\fracpart(\langle\ba,\bx^{(t)}\rangle) \;, 
\]
which makes the error probability $\frac{1}{\poly(s)}$ at most. This means that we have
\[
I(\fracpart(\langle\ba,\bx^{(1)}\rangle_,\ldots,\fracpart(\langle\ba,\bx^{(t)}\rangle); x_j)=\Omega\left(\frac{\log s}{s}\right)
\]

\paragraph{Mutual information from a single instance.}
On the other hand, by the chain rule for mutual information, i.e., \thmref{thm:chain:information}, we have 
\begin{align*}
I&(\fracpart(\langle\ba,\bx^{(1)}\rangle),\ldots,\fracpart(\langle\ba,\bx^{(t)}\rangle); x_j)\\
&=\sum_{k=1}^t I(\fracpart(\langle\ba,\bx^{(k)}\rangle); x_j\,\mid\,\fracpart(\langle\ba,\bx^{(1)}\rangle),\ldots,\fracpart(\langle\ba,\bx^{(k-1)}\rangle)).   
\end{align*}
Since $\fracpart(\langle\ba,\bx^{(k)}\rangle)$ is independent of $\fracpart(\langle\ba,\bx^{(1)}\rangle,\ldots\langle\ba,\bx^{(k-1)}\rangle)$ conditioned on $x_j$, we have
\begin{align*}
\Omega(1)=I(\fracpart(\langle\ba,\bx^{(1)}\rangle),\ldots,\fracpart(\langle\ba,\bx^{(t)}\rangle); x_j)\le\sum_{k=1}^t I(\fracpart(\langle\ba,\bx^{(k)}\rangle); x_j)
=\sum_{k=1}^t I(\fracpart(\langle\ba,\bx\rangle);x_j).    
\end{align*}
Thus we have 
\[I(\fracpart(\by^\top\bA\bx;x_j))=I(\fracpart(\langle\ba,\bx\rangle);x_j)=\Omega\left(\frac{\log s}{st}\right)=\Omega\left(\frac{1}{s  }\right).\]
\end{proof}

\begin{lemma}
\lemlab{lem:remove:s:heavy:frac}
Let $\bA\in\mathbb{Z}^{r\times n}$ be a fixed matrix. 
There exists a pre-processing procedure to $\bA$ and produces a matrix $\bA'\in\mathbb{Z}^{r\times n}$ 
% with the same row span as $\bA$, for $r'=r+\O{rs^2\log^{5/2}n}$. 
such that $\bA'$ zero out at most $\O{rs \log n}$ columns of $\bA$.
Moreover, we have $|\fracpart(\by^\top\bA')_j|^2 \le \frac{1}{s}\cdot\|\fracpart(\by^\top\bA')\|_2^2$ for all $\by\in\mathbb{R}^{r}$ and $j\in[n]$. 
\end{lemma}
\begin{proof}
Let $\bA^{(0)}=\bA$ and let $\bx\in\{-1,0,1\}^n$ drawn from the same distribution in \lemref{lem:random:vec:info:frac}, so that $\bA\bx$ has $\O{r\log n}$ bits. 
Suppose that the above procedure ends in $T$ rounds where $T \le n$.
% Let $T$ be a fixed number of iterations and 
Specifically, for each $t\in[T]$, we identify a columns $j_t\in[n]$ such that $|\fracpart(\by^\top\bA^{(t-1)})_j|^2 > \frac{1}{s}\cdot\|\ \fracpart(\by^\top\bA^{(t-1)})\|_2^2$ 
and  let $\bA^{(t)}$ be the matrix $\bA^{(t-1)}$ after zeroing out the identified column. 
We next apply the chain rule for mutual information, i.e., \thmref{thm:chain:information}. 
we have
\[I(\bA^{(T)}\bx; x_{j_1}, x_{j_2}, \cdots, x_{j_T})=\sum_{t=1}^T I(\bA^{(T)}\bx; x_{j_t}\,\mid\,x_{j_{t + 1}}, x_{j_{t + 2}}, \cdots x_{j_{T}}) \;. \]
Note that given the matrix $\bA$ and $x_{j_{t + 1}}, x_{j_{t + 2}}, \cdots x_{j_{T}}$, we can recover $\bA^{(t)}$. 
Hence, by a similar approach to that in \lemref{lem:random:vec:info:frac}, we have $I(\bA^{(T)}\bx; x_{j_t}\,\mid\,x_{j_{t + 1}}, x_{j_{t + 2}}, \cdots x_{j_{T}}) \ge \O{\frac{1}{s \log s}}$ from $|\fracpart(\by^\top\bA^{(t-1)})_j|^2 > \frac{1}{s}\cdot\|\ \fracpart(\by^\top\bA^{(t-1)})\|_2^2$.
Since $\bA\bx$ can be represented using $\O{r\log n}$ bits, then it follows that 
$I(\bA^{(T)}\bx; x_{j_1}, x_{j_2}, \cdots, x_{j_T}) \le\O{r \log n}$. 
Putting these two things together, we have that $C r \log n \ge\frac{T}{s \log s}$ for some constant $C$, which means that we have $T = \O{rs \log n \log s}$.

% We now summarize the properties of the resulting matrix $\bA'$ in the following corollary.

% after $T=\O{rs^2\log^{5/2} n}$ iterations, we must have $|(\by^\top\bA^{(T)})_j|^2<\frac{1}{s}\cdot\|\by^\top\bA^{(T)}\|_2^2$. 
% Finally, we set $\bA'$ by appending the elementary rows $\be_j$ for $j\in\cup_{t\in[T]}S_t$, so that $\bA'$ has the same row span as $\bA$. 
% Moreover, since each elementary row $\be_j$ has a single nonzero entry, it follows that $|(\by^\top\bA')_j|^2<\frac{1}{s}\cdot\|\by^\top\bA'\|_2^2$ for all $\by\in\mathbb{R}^{r'}$ and $j\in[n]$, where $r'=r+\O{rs^2\log^{5/2} n}$. 
\end{proof}
% \samson{Statement should instead say that we decompose into sparse and dense part}

% \samson{might want to comment that this includes removing the heavy parts}

% \begin{corollary}
%     For any sketching matrix $\bA \in \mathbb{Z}^{r \times n}$, the above pre-processing procedure produces a new matrix $\bA' \in \mathbb{Z}^{r' \times n}$ for $r' = O(rs \log n + rs \log n \log s)$ satisfying the following properties:
%     \begin{itemize}
%         \item The sparse part $\mathbf{I}$ of $\bA'$ consists of standard basis vectors $\{e_i\}_{i \in S}$ for some subset of indices $S \subset [n]$. 
%         \item For all $\by \in \mathbb{R}^r$, $\by^{\top} \bD = v$ is dense. In other words, any linear combination of dense rows will have at least $s$ non-zero entries.
%         \item $\bA'$ satisfies $|\fracpart(\by^\top\bA')_j|^2 \le \frac{1}{s}\cdot\|\fracpart(\by^\top\bA')\|_2^2$ for all $\by\in\mathbb{R}^{r}$ and $j\in[n]$
%         \item The sets of non-zero indices in the sparse part $\mathbf{I}$ and dense part $\bD$ are disjoint.
%     \end{itemize}
% \end{corollary}

\section{Attack Against Linear Sketches}
\label{sec:attack}
In this section, we give a full description of 
our attack against linear sketches for $\ell_0$-estimation. In particular, we prove the following theorem.
\begin{theorem}
    \thmlab{thm:main-theorem}
    \label{thm:main-theorem}
    Suppose that $\mathcal{A}$ is a linear streaming algorithm that solves the $(\alpha + c, \beta - c)$-$\ell_0$ gap norm problem with some constant $\alpha, \beta$ and $c$, where $\bA \in \mathbb{Z}^{r \times n}$ is the sketching matrix with $r << n$, $f: \mathbb{Z}^{r \times n} \rightarrow \{-1, +1\}$ is any estimator used by $\mathcal{A}$, and $\mathcal{A}$ returns $f(\bA, \bA \bx)$ for each query $\bx$. 
    
    Then, there exists a randomized algorithm, which after making an adaptive sequence of queries to $\mathcal{A}$, with high constant probability can generate a distribution $D$ on $\mathbb{Z}^n$ such that $\mathcal{A}$ fails on $D$ with constant probability. Moreover, this adaptive attack algorithm makes at most $\tilde{\mathcal{O}}(r^{8})$ queries and runs in $\poly(r)$ time. %  \polylog(n)
\end{theorem}

\subsection{Construction and Analysis Overview}

The full description of the attack is given in \figref{fig:1}. We first define the probability distribution $P_{a,b}$ with support $[a,b]$ to have probability density function $\mu(p)=\frac{C_{a,b}}{\sqrt{p(1-p)}}$, where $C_{a,b}$ is a normalizing constant. 

For $p\in[0,1]$, we define $\phi^p:\{\pm1\}\to\mathbb{R}$ by $\phi^0(c)=\phi^1(c)=0$, and for $p\in(0,1)$, $\phi^p(1)=-\sqrt{\frac{p}{1-p}}$ and $\phi^p(-1)=\sqrt{\frac{1-p}{p}}$ so that by construction, $\phi^p(c)$ has mean $0$ and variance $1$ when $\PPr{c=-1}=p$ and $\PPr{c=1}=1-p$. 

%\eg{I think we might need to adjust this slightly: at each step $j$, if $i \in I^{j-1}$ (it was accused before), then we should set $x_i^j = 0$ deterministically in all future rounds. Then, at the very end, we will have accused all of the correct indices $i \in S$ and thus will have zero-ed out all coordinates corresponding to the sparse part, and only the dense part will be left to answer queries (but we know the dense part cannot distinguish between $||x||_0 > 3n/4$ and $||x||_0 < n/4$). We should probably also query the entire vector at each step (including the $0$'s from previously discovered coordinates)?}

%\hl{I think the previous writing is the same as the above way but the notation might be confusing so I define a $z_S(v) $ means we zero the coordinates in $S$. Feel free to make edits if you see something to make it clearer.The current proof does not explicitly say we should find all the $S$ at the end, as at some point if the algorithm already fails after we zero part of $S$, we may not find the entire set of $S$ but our goal of attack is already achieved. So our current argument is intuitively saying as we find more and more entries in $S$, the correctness guarantee of the algorithm can not always holds so we can attack it at some point.} 
\begin{figure}

\begin{mdframed}
\begin{algorithmic}

Let $\alpha$ and $\beta$ be defined as in \lemref{lem:moment:match}
\newline
Let $\mathcal{D}$ be the distribution family in \lemref{lem:moment:match} with $K = \O{r \log n}$
\newline\noindent
$h \gets \O{rs \log n} = \O{r^4 \log^3 n}$, $\sigma\gets\O{h \log(n)}$, $\ell\gets\O{h}\cdot\sigma$, $c \gets \O{1}$
\newline \noindent
Let $z_J(v) $ denote the vector where we make $v_i$  to $0$ for all $i \in J$.
\newline$\mathcal{A} \gets$ An instantiation of the $\ell_0$ gap-norm algorithm. %-
\newline\noindent
Initialize $s_i^0=0$ for all $i\in[n]$
\newline \noindent 
For $j\in[\ell]$:
\newline \indent
Sample $u^1, \cdots, u^c \sim D_\alpha^n$ and $v^1, \cdots, v^c \sim D_{\beta}^n$ % v^1, \cdots v^c 
% \newline \indent
% Set $u^i \gets u^i_{I^j}$ and $v^i \gets v^i_{I^j}$ where we zero the coordinates on $I^j$% let
\newline \indent
If $\mathcal{A}$ fails with constant probability on one of $z_{I^{j - 1}}(u^i)$ or $z_{I^{j - 1}}(v^i)$:
\newline \indent \indent Output this distribution as the attack.
\newline\indent
Sample $p^j\sim P_{\alpha,\beta}$ and $v^j\sim D_{p^{j}}^n$
% \newline \indent
% Set $v^i \gets v^i_{I^j}$.
\newline\indent
For all $i\in[n]$, set $c_i^j=1$ if $v_i^j\neq 0$ and $c_i^j=-1$ otherwise if $v_i^j=0$
\newline\indent
Query $z_{I^{j - 1}}(v^j)\in\mathbb{Z}^n$ and receive $a^j = \mathcal{A}(z_{I^{j - 1}}(v^j))\in\{\pm1\}$ as the output % v^j_{[n] \setminus I^{j - 1}}
\newline\indent
For $i\in[n]$, update $s_i^j\gets s_i^{j-1}+a^j\cdot\phi^{p^j}(c_i^j)$
\newline\indent
Set $I^j= I^{j - 1} \cup \{i\in[n]\,\mid\, s_i^j>\sigma\}$ and $\mathcal{S}^{j + 1} = \mathcal{S} \setminus I^j$    
\end{algorithmic}
\end{mdframed}   
\caption{Construction of Our Attack}
\figlab{fig:1}
\end{figure} 

In this section, we give a high-level description of our algorithm. First, recall that by \lemref{lem:pre} with parameter $s = \O{r^3 \log^3 n}$, we can assume the sketching matrix $\bA$ has the form
\[
\bA = \begin{bmatrix}
        \bD \\
        \bS
    \end{bmatrix}
\]

without loss of generality in our attack. Importantly, we recall that $\bA$ has the following properties:

\begin{enumerate}
    \item The $\bD$ and $\bS$ are column-disjoint.
    \item We have that $|\fracpart(\by^\top\bD)_j|^2 \le \frac{1}{s}\cdot\|\fracpart(\by^\top\bD)\|_2^2$ for all $\by\in\mathbb{R}^{r}$ and $j\in[n]$ 
    \item $\bS$ contains at most $h = \O{r s \log n} = \O{r^4 \log^3 n}$ non-zero columns.
\end{enumerate}
Suppose that $\mathcal{D}$ is the distribution family in \lemref{lem:moment:match} with $K = \O{r \log n}$. Then, for $\bx \sim D_\alpha$ and $\bx' \sim D_\beta$, by \lemref{lem:distribution:tvd}, we know that $d_{\mathrm{tv}}(\bD \bx, \bD \bx') \le \frac{1}{\poly(n)}$. Let $\mathcal{S}$ denote the set of non-zero column indices of $\bS$. Then, if we set $\bx_i=0$  and $\bx_i'=0$ for all $i \in \mathcal{S}$, the streaming algorithm $\mathcal{A}$ must fail to solve the $\ell_0$ gap norm problem in one of these two cases, with constant probability. With this motivation in mind, the main task of the adaptive adversary is to design an adaptive sequence of queries to learn the set of indices in $\mathcal{S}$.

 In each iteration, we query a random vector $\bx^t \sim D_{p}^n$ where $p$ is sampled in $P_{\alpha, \beta}$. 
We maintain a score $s_i^t$ for each coordinate $i \in [n]$, which represents some measure of the \textit{correlation} between the $i$-th coordinate of the inputs and the outputs of the algorithm $\mathcal{A}$ up until step $t$. In particular, at the $t$-th iteration and for each coordinate, let $c_i^t = 1$ if $x_i^t \ne 0$ and $c_i^t = -1$ if $x_i^t = 0$. Suppose that the output of the algorithm is $a^t$, then we update each coordinate's current store $s_i^t = s_i^{t - 1} + a^t \cdot \phi^p(c_i^t)$, where $\phi^p(\cdot)$ is some specially-chosen function which depends on the choice of $p$. 
If for coordinates $i$, the score $s_i^t$ exceeds a pre-determined threshold $\sigma$, then we accuse this coordinate and treat it as a coordinate in the secret set $\mathcal{S}$. Furthermore, we set this coordinate to $0$ in all future queries, so the algorithm cannot get any information about this coordinate in future iterations. 

 From the guarantee of the algorithm $\mathcal{A}$, we get that when $p$ is close to $\alpha$, with high probability it should output $-1$ and when $p$ is close to $\beta$, it should output $1$. However, from \lemref{lem:distribution:tvd} we know that the algorithm $\mathcal{A}$ can almost get nothing about the value of $p$ from the part of the sketch $\bD \bx_D$, which means its output should have a higher correlation with the coordinates in $\mathcal{S}$. With this ind mind, the proof is comprised of the following two parts:
 \begin{itemize}
     \item \textbf{Soundness:} For any coordinate $i \not \in \mathcal{S}$,  the score $s_i$ will never exceed $\sigma$ with high probability, which means that coordinate $i$ will never be falsely accused.
     \item \textbf{Completeness:} Let $\mathcal{S}^j$ be the remaining (undiscovered) coordinates in $\mathcal{S}$ in the $j$-th round. We will show that if the algorithm still has the correctness guarantee on $D_\alpha^n$ and $D_\beta^n$ after we zero out the accused coordinates, then the sum $\sum_{i \in \mathcal{S}} s_i$ will increase faster than the scores of other coordinates. This means that as we accuse more and more coordinates in $\mathcal{S}$, we either find a distribution that $\mathcal{A}$ or find more coordinates in $\mathcal{S}$ (note that if find the whole set $\mathcal{S}$, the algorithm $\mathcal{A}$ must fail in the next iteration).
 \end{itemize}

 To simplify our argument in Sections \ref{sec:soundness} and \ref{sec:completeness}, we will first prove that soundness and completeness hold in the case that the streaming algorithm $\mathcal{A}$ only uses $\bx_S$ to estimate the $\ell_0$ norm at each step. Then, in Section~\ref{sec:main-theorem}, we show that the soundness and completeness guarantees hold for an arbitrary algorithm $\mathcal{A}$ which uses both the sparse part $\bx_S$ and the dense part $\bD x_D$ to compute its responses to each query.

 % estimator $f$ is only taking the input of $\bS\bx_\mathcal{S}$, and it randomly generates an input $\bx_D \sim D_{\gamma}^{|D|}|$ for a fixed $\gamma \in [\alpha, \beta]$ to compute a $\bD \bx_{D}$ as another part of the input vector. That is, the estimator $f$ is decided only by the coordinates that correspond to $\mathcal{S}$. 
 % Then in Section~\ref{sec:main-theorem} we show that since the total variance 
 
\subsection{Soundness}
\label{sec:soundness}

\begin{lemma}
\lemlab{lem:one:moment:gen:subg}
For $p\in[\alpha, \beta]$, let $\tau=\min(\alpha,1 - \beta)$ and $t\in\left[-\frac{\sqrt{\tau}}{2},\frac{\sqrt{\tau}}{2}\right]$. 
Then 
\[\EEx{v\sim D_p}{e^{t\phi^p(c)}}\le e^{t^2},\] 
where $c=1$ if $v$ is nonzero and $c=-1$ if $v$ is zero. 
\end{lemma}
\begin{proof}
Although the statement is slightly different from Lemma 2.4 in \cite{SteinkeU15} due to drawing $v\sim D_p$, the proof is almost verbatim; we include it for completeness. 

Observe that $\EEx{v\sim D_p}{\phi^p(c)}=p\cdot\phi^p(-1)+(1-p)\cdot\phi^p(1)$, since $\PPPr{v\sim D_p}{v=0}=p$. 
Then we have $\EEx{v\sim D_p}{\phi^p(c)}=0$ and $\EEx{v\sim D_p}{(\phi^p(c))^2}=1$. 
Moreover, for $c\in\{\pm1\}$, we have $|\phi^p(c)|\le\frac{1}{\sqrt{\tau}}$. 
Thus, we have $|\phi^p(c)\cdot t|\le\frac{1}{2}$. 

Since $e^x\le 1+x+x^2$ for $x\in\left[-\frac{1}{2},\frac{1}{2}\right]$, then
\[\EEx{v\sim D_p}{e^{t\phi^p(c)}}\le1+t\cdot\EEx{v\sim D_p}{\phi^p(c)}+t^2\cdot\EEx{v\sim D_p}{(\phi^p(c))^2}=1+t^2\le e^{t^2}.\] 
\end{proof}

\begin{lemma}
\lemlab{lem:sum:moment:gen:subg}
Let $p^1,\ldots,p^m\in[\alpha,\beta]$ and $v_i\sim D_{p^j}$. 
Let $a^1,\ldots,a^m\in[-1,1]$ be fixed and $\tau=\min(\alpha, 1 - \beta)$. 
Then for all $\lambda\ge 0$,
\[\PPr{\sum_{j\in[m]}a^j\phi^{p^j}(c_i^j)\ge\lambda}\le e^{-\lambda^2/4m}+e^{-\sqrt{\tau}\lambda/4},\]
where for all $i\in[n]$, we have $c_i=1$ if $v_i^j\neq 0$ and $c_i^j=-1$ otherwise if $v_i^j=0$. 
\end{lemma}
\begin{proof}
The proof follows exactly along the lines of Lemma 2.5 in \cite{SteinkeU15}, using $\tau=\min(\alpha, 1 - \beta)$ instead due to the range of $p\in[\alpha,\beta]$ as the probability of $D_p$ drawing a zero.  
We include the proof for completeness. 
By \lemref{lem:one:moment:gen:subg}, for all $t\in\left[-\frac{\sqrt{\tau}}{2},\frac{\sqrt{\tau}}{2}\right]$, 
\[\EEx{v}{e^{t\sum_{i\in[m]}a^j\phi^{p^j}(c_i^j)}}\le\prod_{j\in[m]}\EEx{v_i^j\sim D_{p^j}}{e^{ta^j\phi^{p^j}(c_i^j)}}\le e^{t^2m}.\]
By Markov's inequality, we have
\[\PPr{\sum_{i\in[m]}a^j\phi^{p^j}(c_i^j)\ge\lambda}\le\frac{\Ex{e^{t\sum_{i\in[m]}a^j\phi^{p^j}(c_i^j)}}}{e^{t\lambda}}\le e^{t^2m-t\lambda}.\]
We set $t=\min\left(\frac{\sqrt{\tau}}{2},\frac{\lambda}{2m}\right)$. 
Then for $\lambda\in\left[0,m\sqrt{\tau}\right]$, we have $t=\frac{\lambda}{2m}$ and so
\[\PPr{\sum_{i\in[m]}a^j\phi^{p^j}(c_i^j)\ge\lambda}\le e^{-\lambda^2/4m},\]
and for 
$\lambda\ge m\sqrt{\tau}$,
\[\PPr{\sum_{i\in[m]}a^j\phi^{p^j}(c_i^j)\ge\lambda}\le e^{\tau m/4-\sqrt{\tau}\lambda/2}\le e^{-\frac{\sqrt{\tau}\lambda}{4}}.\]
\end{proof}

\begin{theorem}[Etemadi's inequality]
\cite{etemadi1985some}
\thmlab{thm:etemadi}
Let $X_1,\ldots,X_n$ be independent random variables and for all $k\in[n]$, let $S_k=\sum_{i=1}^k X_i$ be the $k$-th partial sum of the sequence $X_1,\ldots,X_n$. 
Then for all $\lambda>0$,
\[\PPr{\max_{k\in[n]}|S_k|>4\lambda}\le4\cdot\max_{k\in[n]}\PPr{|S_k|>\lambda}.\]
\end{theorem}

\begin{lemma}[Individual soundness]
\lemlab{lem:ind:sound} 
% Let $I = \cup_{j \in [\ell]} I^j$
For all $i\in[n]\setminus\calS$, we have
\[\PPr{i\in I^\ell}\le \frac{1}{n^2}.\]
\end{lemma}
\begin{proof}
The proof follows similarly from Proposition 2.7 in \cite{SteinkeU15}.
Consider a fixed $i\in[n]\setminus\calS$. 
Since the adversary does not see $c_i^j$, without the loss of generality we can assume that the outputs $a^j$ are fixed and $v_i^j$ and then $c_i^j$ are subsequently drawn since $i \notin \mathcal{S}$. 
By \lemref{lem:sum:moment:gen:subg}, we have for every $j\in[\ell]$,
\begin{align*}  
\PPr{s_i^j>\frac{\sigma}{4}} &=\PPr{\sum_{k\in[j]} a^k\phi^{p_k}(c_i^k)>\frac{\sigma}{4}} \le e^{-\frac{\sigma^2}{64\ell}}+e^{-\sigma\sqrt{\tau}/16}
\end{align*}

Similarly, by \lemref{lem:sum:moment:gen:subg}, for every $j\in[\ell]$,

\begin{align*}
\PPr{s_i^j<-\frac{\sigma}{4}} &=\PPr{\sum_{k\in[j]} a^k\phi^{p_k}(c_i^k)<-\frac{\sigma}{4}} \le e^{-\frac{\sigma^2}{64\ell}}+e^{-\sigma\sqrt{\tau}/16}
\end{align*}

Thus by \thmref{thm:etemadi}, 
\begin{align*}
\PPr{i\in I^j}&\le\PPr{\max_{t\in[j]}|s_i^t|>\sigma}\\
&\le4\max_{t\in[j]}\PPr{|s_i^t|>\frac{\sigma}{4}}\\
&\le8(e^{-\frac{\sigma^2}{64\ell}}+e^{-\sigma\sqrt{\tau}/16}) \le \frac{1}{n^2}.   
\end{align*}

% From the choice of our parameter, we have 
% \begin{align*}
%     \PPr{i\in I^j}& = \PPr{|s_i^j| > \sigma} \le \frac{1}{\poly(n)}.
% \end{align*}
% Taking a union bound over $j = 1, 2, \cdots, \ell$, we have that 
% \[
% \PPr{i\in I^\ell}\le \frac{1}{n^2}.
% \]
\end{proof}

\begin{lemma}[Soundness] 
\lemlab{soundness}
\[\PPr{|I^\ell\setminus\calS|\ge 1}\le \frac{1}{n}.\]
\end{lemma}
\begin{proof}
The proof follows similarly from Theorem 2.8 in \cite{SteinkeU15}, as follows. 
For $i\in[n]\setminus\calS$, we use $Y_i$ to denote the indicator random variable for the event $i\in I^\ell\setminus\calS$. 
% Note that $\{Y_i\,\mid\,i\in[n]\setminus\calS\}$ are independent conditioned on the value of $\calS$ and the choice of $p_j$ for $j\in[\ell]$. 
By \lemref{lem:ind:sound}, we have that $\Ex{Y_i}\le \frac{1}{n^2}$ for all $i\in[n]\setminus\calS$. 
By Markov's inequality,
\[\PPr{|I^\ell\setminus\calS|\ge 1}\le\Ex{\sum_{i\in[n]\setminus\calS}Y_i}\le \frac{1}{n^2}(n-r)\le \frac{1}{n}. \qedhere \]
\end{proof}

\begin{lemma}
\lemlab{lem:soundness:bound}
For each $i\in[n]$, let $j_i\in[\ell+1]$ be the first $j$ such that $i\in I^j$, where we set $I^{\ell+1}=[n]$. 
Then for $\tau=\min(\alpha,\beta)$ and for any $J\subset[n]$,
\[\PPr{\sum_{i\in J}(s_i^\ell-s_i^{j_i-1})>\lambda})\le e^{-\frac{\lambda^2}{4|J|\ell}}+e^{-\sqrt{\tau}\lambda/4}.\]
\end{lemma}
\begin{proof}
The proof is nearly identical to that of Lemma 2.10 in \cite{SteinkeU15}, as follows. 
Observe that
\[\sum_{i\in J}(s_i^\ell-s_i^{j_i-1})=\sum_{i\in J}\sum_{j\in[\ell]}\mathbb{I}(j\ge j_i) a^j\phi^{p_j}(c_i^j),\]
where $\mathbb{I}$ denotes the standard indicator function. 
Again, since we zero out the $i$-th coordinate after time $j_t - 1$, we can take the view that the outputs $a^j$ are fixed and then the terms $\phi^{p_j}(c_i^j)$ are subsequently drawn for $j \ge j_t$. 
% In particular we can take the view that $\mathbb{I}(j\ge j_i)\cdot a_j\in[-1,1]$ is fixed and then $\phi^{p_j}(c_i^j)$ is subsequently drawn. 
Then by \lemref{lem:sum:moment:gen:subg}, we have
\[\PPr{\sum_{i\in J}(s_i^\ell-s_i^{j_i-1})>\lambda}\le e^{-\frac{\lambda^2}{4|J|\ell}}+e^{-\sqrt{\tau}\lambda/4},\]
as desired.
\end{proof}

\subsection{Completeness}
\label{sec:completeness}

% Throughout this section, we assume that the event in \lemref{soundness} happens. That is, no coordinates in $[n] \setminus \mathcal{S}$ were accused during $\ell$ iterations. 
Recall that we use $h$ to denote the size of $\mathcal{S}$, where $\mathcal{S}$ corresponds to the non-zero indices of columns in $\bS$.

\subsubsection{Fourier Analysis}
\begin{lemma}
\lemlab{lem:corr:fourier}
Let $f:\mathbb{R}^h\to\mathbb{R}$ and let $g:[0,1]\to\mathbb{R}$ be defined so that $g(p)=\EEx{v_1,\ldots,v_h\sim D_p}{f(v)}$, where for all $i\in[h]$, $c_i=1$ if $v_i$ is nonzero and $c_i=-1$ if $v_i$ is zero. 
Then for any $p\in[\alpha,\beta]$,
\[\EEx{v_1,\ldots,v_h\sim D_p}{f(v)\cdot\sum_{i\in[h]}\phi^p(c_i)}=g'(p)\sqrt{p(1-p)}\]
\end{lemma}
\begin{proof}
The analysis is similar to Lemma 2.11 of \cite{SteinkeU15}, but with differing probability distributions and thus correspondingly slightly differing score functions $\phi$. 
We include the full proof for completeness.

For $p\in(0,1)$ and $T\subset[h]$, we define $\phi_T^p:\{\pm1\}^h\to\mathbb{R}$ by $\phi_T^p(c)=\prod_{i\in T}\phi^p(c_i)$, so that the functions $\phi_T^p$ form an orthonormal basis with respect to the product distribution with bias $p$. 
Specifically, we have that for all $T,U\subset[h]$,
\[\EEx{v_1,\ldots,v_h\sim D_p}{\phi_T^p(c)\cdot\phi_U^p(c)}=\begin{cases}1\qquad&T=U\\
0\qquad &T\neq U,\end{cases}\]
where for all $i\in[h]$, $c_i=1$ if $v_i$ is nonzero and $c_i=-1$ if $v_i$ is zero. 
We use $c(v)$ to denote this mapping from $v$ to $c$. 
Therefore, we can decompose $f$ by
\[f(v)=\sum_{T\subset[h]}\widehat{f}^p(s)\cdot\phi_T^p(c(v)),\]
where we have the Fourier coefficients
\[\widehat{f}^p(T)=\EEx{v_1,\ldots,v_h\sim D_p}{f(v)\cdot\phi_T^p(c(v))},\]
where all $T\subset[h]$. 
Then for $p,q\in(0,1)$, we can expand $g(q)$ by
\begin{align*}
g(q)&=\EEx{v_1,\ldots,v_h\sim D_q}{f(v)}\\
&=\sum_{T\subset[h]}\widehat{f}^p(T)\cdot\EEx{v_1,\ldots,v_h\sim D_q}{\phi_T^p(c(v))}\\
&=\sum_{T\subset[h]}\widehat{f}^p(T)\cdot\prod_{i\in T}\EEx{v_i\sim D_q}{\phi_T^p(c_i(v_i))}\\
&=\sum_{T\subset[h]}\widehat{f}^p(T)\cdot\left(q\cdot\sqrt{\frac{1-p}{p}}+(1-q)\cdot\sqrt{\frac{p}{1-p}}\right)^{|T|},
\end{align*}
since for $v_i\sim D_q$, the probability that $v_i$ is zero (and thus $c_i$ is $-1$) is $q$. 
Thus, we have
\[g'(q)=\sum_{T\subset[h], T\neq\emptyset}\widehat{f}^p(T)\cdot|T|\cdot\left(q\cdot\sqrt{\frac{1-p}{p}}+(1-q)\cdot\sqrt{\frac{p}{1-p}}\right)^{|T|-1}\cdot\left(\sqrt{\frac{1-p}{p}}+\sqrt{\frac{p}{1-p}}\right)\]
and
\[g'(p)=\sum_{i\in[h]}\widehat{f}^p(\{i\})\cdot\left(\sqrt{\frac{1-p}{p}}+\sqrt{\frac{p}{1-p}}\right).\]
Since $\widehat{f}^p(\{i\})=\EEx{v_1,\ldots,v_h\sim D_p}{f(v)\cdot\phi^p(c_i(v_i))}$, then
\[\EEx{v_1,\ldots,v_h\sim D_p}{f(v)\cdot\sum_{i\in[h]}\phi^p(c_i(v_i))}=\sum_{i\in[h]}\widehat{f}^p(\{i\})=\frac{g'(p)}{\sqrt{\frac{1-p}{p}}+\sqrt{\frac{p}{1-p}}}=g'(p)\cdot\sqrt{p(1-p)}.\]
\end{proof}

\begin{lemma}
\lemlab{lem:corr:gap}
Let $f:\mathbb{R}^h\to\mathbb{R}$ and let $g:[0,1]\to\mathbb{R}$ be defined so that $g(p)=\EEx{v_1,\ldots,v_h\sim D_p}{f(v)}$, where for all $i\in[h]$, $c_i=1$ if $v_i$ is nonzero and $c_i=-1$ if $v_i$ is zero. 
Then there exists a constant $\zeta>0$ such that
\[\EEx{p\sim P_{\alpha,\beta}}{\EEx{v_1,\ldots,v_h\sim D_p}{f(v)\cdot\sum_{i\in[h]}\phi^p(c_i)}}\ge\zeta\cdot(g(\beta)-g(\alpha)).\]
\end{lemma}
\begin{proof}
The proof follows similarly from Proposition 2.12 in \cite{SteinkeU15}, as follows. 
Recall that $P_{\alpha,\beta}$ has probability density function $\mu(p)=\frac{C_{\alpha,\beta}}{\sqrt{p(1-p)}}$ entirely on the interval $[\alpha,\beta]$. 
By \lemref{lem:corr:fourier}, 
\begin{align*}
\EEx{p\sim P_{\alpha,\beta}}{\EEx{v_1,\ldots,v_h\sim D_p}{f(v)\cdot\sum_{i\in[h]}\phi^p(c_i)}}&=\EEx{p\sim P_{\alpha,\beta}}{g'(p)\cdot\sqrt{p(1-p)}}\\
&=\int_{\alpha}^\beta g'(p)\sqrt{p(1-p)}\cdot\mu(p)\,dp\\
&=C_{\alpha,\beta}\cdot\int_{\alpha}^\beta g'(p)\,dp\\
&=C_{\alpha,\beta}\cdot(g(\beta)-g(\alpha)).
\end{align*}
The proof then follows from setting $C_{\alpha,\beta}=\zeta$. 
\end{proof}

\subsubsection{Concentration}

\begin{lemma}
\lemlab{lem:function:gap}
% For the $j \in [\ell]$ round, if $|\mathcal{S}^j| \le O(\log r)$, then the $\ell_0$ estimation algorithm $\mathsf{ALG}$ can not have error probability less than $1/r$ when sample from the specific distribution(we zero the coordinates in $I^j)$. On the other hand, if 
% the $\ell_0$ estimation algorithm $\mathsf{ALG}$ 

Suppose that at the $j$-th round, the algorithm $\mathcal{A}$ has error probability $\delta_\alpha^j, \delta_\beta^j \le c$ over the input distribution $z_{I^{j - 1}}(u)$ and $z_{I^{j - 1}}(v)$ where $u \in D_\alpha^n, v \in D_\beta^n$ for some small constant $c$, then we have there exists a function $f^{j}: \mathbb{R}^h \to \mathbb{R}$ that only depends
on the interaction up to round $j - 1$ and satisfies  the value of $f^j(v)$ is decided by the coordinates of $v$ in $\mathcal{S}^j$ and $f^j(v^j_{\mathcal{S}^j}) = a^j$ where $\mathcal{S}^{j} = \mathcal{S} \setminus I^{j - 1}$. Moreover, we have
and
$g^j(\beta) - g^j(\alpha) \ge 2 - \eta$ for some $\eta = \O{1}$. % then we must have 
\end{lemma}

\begin{proof}
From the assumption we have 
\[\EEx{v_1, \ldots, v_n \sim D_\alpha}{\mathcal{A}(z_{I^{j - 1}}(v))} \le -(1 - 2c)\]
and 
\[\EEx{v_1, \ldots, v_n \sim D_\beta}{\mathcal{A}(z_{I^{j - 1}}(v))} \ge 1 - 2c.\]
From the fact that we have zeroed out the coordinates of $v$ in $I^{j - 1}$, we can without loss of generality get that the output $a^j$ can be represented by the value of a function $f^j$ that is only decided by the coordinates on $\mathcal{S}^j$. From this and the assumption that $\delta_\alpha^j, \delta_\beta^j \le c$ we have that 
    \[
    g^j(\beta) - g^j(\alpha) =  
 \EEx{v_1,\ldots,v_h\sim D_\beta}{f(v)} - \EEx{v_1,\ldots,v_h\sim D_\alpha}{f(v)} \ge 2 - \eta
    \]
    for some $\eta = \O{1}$.
\end{proof}

For $p\sim P_{\alpha,\beta}$ and $v_1,\ldots,v_h\sim D_p$, we set $\xi_{\alpha,\beta}(f)=f(v)\cdot\sum_{i \in [h]}\phi^p(c_i)$, where for all $i\in[h]$, $c_i=1$ if $v_i$ is nonzero and $c_i=-1$ if $v_i$ is zero. 

\begin{lemma}
\lemlab{lem:exp:bound}
Let $f:\mathbb{R}^h\to\{\pm1\}$, $\tau=\min(\alpha, 1 - \beta)$ and $t\in\left[-\frac{\sqrt{\tau}}{8},\frac{\sqrt{\tau}}{8}\right]$. 
Then for $C=\frac{32e^{h\tau/64}}{\tau}$, we have
\[\Ex{e^{t\xi_{\alpha,\beta}(f)-\Ex{t\xi_{\alpha,\beta}(f)}}}\le e^{Ct^2}.\]
\end{lemma}
\begin{proof}
Let $Y = \sum_{i\in[h]}\phi^p(c_i)$. By \lemref{lem:one:moment:gen:subg} and independence, we have that 
\begin{align*}
    \Ex{e^{tY}} = \Ex{e^{t \sum_{i\in[n]}\phi^p(c_i)}} = \left(\EEx{v \sim D_p} {e^{t\phi^p(c)}}\right)^h\le e^{{t^2 h}} % \mathbb{E} {v_1, v_2, ..., v_n \sim D_p}
\end{align*}
for $t \in \left[-\frac{\sqrt{\tau}}{8},\frac{\sqrt{\tau}}{8}\right]$. Pick $t \in \{\pm \frac{\sqrt{\tau}}{8}\}$ such that 
\[
\sum\limits_{k = 0}^{\infty} \frac{t^{2k + 1}}{(2k + 1)!} \Ex{Y^{2k + 1}} \ge 0.
\]
Then by dropping the positive terms, for all $j \ge 1$, we have 
\begin{align*}
    0 \le \Ex{Y^{2j}} \le \frac{(2j)!}{t^{2j}} \sum_{k = 0}^{\infty} \frac{t^k}{k!} \Ex{Y^k} = \frac{(2j)!}{t^{2j}} \Ex {e^{tY}} \le \frac{(2j)!}{t^{2j}} e^{t^2 h } = \frac{8^j (2j)!}{\tau^j} e^{\tau h/64} \;.
\end{align*}
This means that we have bounded the even moment of $Y$. For $k = 2j + 1 \ge 3$, by Cauchy-Schwartz, 
\begin{align*}
    \Ex {|Y|^k} \le \sqrt{\Ex{Y^{2j}} \cdot \Ex{Y^{2j + 2}}} \le \sqrt{\frac{8^j (2j)!}{\tau^j}e^{h\tau /64}\frac{8^{j + 1} (2j + 2)!}{\tau^{j + 1}} e^{h\tau /64}} = \frac{8^{k / 2} k! }{\tau^{k / 2}} e^{h\tau / 64} \sqrt{\frac{k + 1}{k}}. 
\end{align*}
Since $|f(c) | \le 1$, we have $\Ex{|f(c) \cdot Y|^k} \le \Ex{|Y|^k} \le 2\cdot 8^{k / 2} k! e^{h\tau / 64} / \tau^{k / 2}$. 

For $t \in \left[-\frac{\sqrt{\tau}}{8},\frac{\sqrt{\tau}}{8}\right]$, we have 
\begin{align*}
    \Ex{e^{t\xi_{\alpha,\beta}(f)}} & \le 1 + t\Ex{\xi_{\alpha,\beta}(f)} + \sum_{k = 2}^{\infty} \frac{|t|^k}{k!} \Ex{|\xi_{\alpha,\beta}(f)|^k} \\
    & \le 1 + t\Ex{\xi_{\alpha,\beta}(f)} + \sum_{k = 2}^{\infty} \frac{|t|^k}{k!} \frac{2 \cdot 8^{k / 2}k!e^{h\tau / 64}}{\tau^{k / 2}} \\
    &= 1 + t\Ex{\xi_{\alpha,\beta}(f)} + 2 e^{h\tau/64}\sum_{k = 2}^{\infty} \left(\frac{\sqrt{8}t}{\sqrt{\tau}}\right)^k \\
    & \le 1 + t\Ex{\xi_{\alpha,\beta}(f)} + 2 e^{h\tau/64}\sum_{k = 2}^{\infty} \left(\frac{\sqrt{8}t}{\sqrt{\tau}}\right)^2 (\sqrt{8})^{-(k - 2)} \\
    & \le 1 + t\Ex{\xi_{\alpha,\beta}(f)} + 32 e^{h\tau/64} \frac{t^2}{\tau} \\
    & \le e^{t\Ex{\xi_{\alpha,\beta}(f)} + Ct^2} \;.
\end{align*}
\end{proof}

\begin{theorem}[Azuma-Doob Inequality, Theorem 2.16 in \cite{SteinkeU15}]
Let $X_1,\ldots,X_m\in\mathbb{R}$, $\mu_1,\ldots,\mu_m\in\mathbb{R}$, and $\calU_0,\ldots,\calU_m\in\Omega$ be random variables such that for all $i\in[m]$:
\begin{itemize}
\item
$X_i$ and $\calU_{i-1}$ are fixed by $\calU_i$
\item 
$\mu_i$ is fixed by $\calU_{i-1}$.
\end{itemize}
Suppose that for all $i\in[m]$, $u\in\Omega$, and $t\in[-c,c]$, we have
\[\Ex{e^{t(X_i-\mu_i)}\,\mid\calU_{i-1}=u}\le e^{Ct^2}.\]
Then for $\lambda\in[0,2Cmc]$, we have
\[\PPr{\left\lvert\sum_{i\in[m]}(X_i-\mu_i)\right\rvert\ge\lambda}\le2e^{-\frac{\lambda^2}{4Cm}},\]
and for $\lambda\ge2Cmc$, we have
\[\PPr{\left\lvert\sum_{i\in[m]}(X_i-\mu_i)\right\rvert\ge\lambda}\le2e^{-\frac{-c\lambda}{2}}.\]
\end{theorem}

\subsubsection{Lower Bounding the Correlation}

\begin{lemma}
\lemlab{lem:completeness:gap}
Let $\tau=\min(\alpha, 1 - \beta)$ and let $\zeta$ be the constant from \lemref{lem:corr:gap}. 
Suppose that for every $j \in [\ell]$-th round, the algorithm $\mathcal{A}$ has error probability $\delta_\alpha^j, \delta_\beta^j \le c$ over the distribution $z_{I^{j - 1}}(D_\alpha^n), z_{I^{j - 1}}(D_\beta^n)$ for some small constant $c$, where $z_{I^{j - 1}}$ means we zero out the coordinates in $I^{j - 1}$.
% Suppose that %we have $|\mathcal{S} \setminus I^\ell|$
% at each round, when we sample from the distribution, the error probability of $\mathsf{ALG}$ is less than some constant $c$.%, then 
Then for any $\lambda\in\left[0,\frac{15\ell}{\sqrt{\tau}}\right]$,
\[\PPr{\sum_{i\in\calS} s_i^\ell<2\ell\zeta(1 -\eta) -\lambda}<2e^{-\frac{\lambda^2\tau}{200\ell}}.\]
\end{lemma}
\begin{proof}
For each $j \in [\ell]$, from the discussion of ~\lemref{lem:function:gap} we can have a function $f^j: \mathbb{R}^h \to \{\pm 1\}$ that only depends
on the interaction up to round $j - 1$ and satisfies $f^j(v^j_{S^j}) = a^j$. Define  % \{\pm 1\}
\[
X_j = f^j(v^j_{S^j}) \sum_{i\in[h]}\phi^p(c^j_i) \sim  \xi_{\alpha,\beta}(f^j) \;,
\]
where $\sim$ denotes that has the same distribution. We than have 
\[
\sum_{i \in \mathcal{S}} s^\ell_i = \sum_{j \in [\ell]} X_j \;.
\]
From \lemref{lem:corr:gap} and \lemref{lem:function:gap}. We have that 
\[
\mu_j = \Ex{X_j} \ge 2 \zeta (1 - \eta)
\]
for all $f^j$. Then, from \lemref{lem:exp:bound} we have, 
\[
\Ex{e^{t(X^j - \mu_j)}} = \Ex{e^{t\xi_{\alpha,\beta}(f^j)-\Ex{t\xi_{\alpha,\beta}(f^j)}}} \le e^{Ct^2} \;.
\]
Define $\mathcal{U}_j = (f^1,p^1,v^1, \cdots, f^j, p^j,v^j,f^{j+1})$. Now $X_1, ..., X_\ell$, $\mu_1, ..., \mu_\ell$, and $\mathcal{U}_1, ..., \mathcal{U}_\ell$ satisfies the condition in \lemref{lem:exp:bound}. For $\lambda \in [0, 2Cmc] = [0, 15\ell / \sqrt{\tau}]$, we have that 
\begin{align*}
   \PPr{\sum_{i\in\calS} s_i^\ell<2\ell\zeta(1 -\eta) -\lambda} \le \PPr{\left|\sum_i X_i - \mu_i\right|} \le 2e^{-\lambda^2 / 4Cm}<2e^{-\frac{\lambda^2\tau}{200\ell}}. 
\end{align*}
\end{proof}

\begin{lemma}
\lemlab{lem:completeness:upper}
Let $\tau=\min(\alpha, 1 -  \beta)$. 
Then for all $\lambda>0$,
\[\PPr{\sum_{i\in\calS}s_i^\ell>\lambda+h\sigma+\frac{h}{\sqrt{\tau}}}\le e^{-\frac{\lambda^2}{4h\ell}}+e^{-\sqrt{\tau}\lambda/4}.\]
\end{lemma}
\begin{proof}
For each $i \in [n]$, let $j_i$ be as in \lemref{lem:soundness:bound}. That is, $i \notin \mathcal{S}^{j_i}$ and $i \in \mathcal{S}^{j_{i - 1}}$, where we define $\mathcal{S}^{\ell + 1} = \emptyset$ and $\mathcal{S}^0 = [n]$. By the definition of $j_i$, we have that $s_i^{j_i - 2} \le \sigma$ for all $i \in \mathcal{S}$. Hence we have 
\[
\sum_{i \in \mathcal{S}} s_i^{j_i - 1} = \sum_{i \in \mathcal{S}} s_i^{j_i - 2} + a^{j_i - 1} \phi^{j_i - 1}(c_i^{j_i - 1}) \le \sum_{i \in \mathcal{S}} (\sigma + \frac{1}{\sqrt{\tau}}) \le h \sigma + \frac{h}{\sqrt{\tau}} \;.
\]
By \lemref{lem:soundness:bound} we have 
\[
\PPr{\sum_{i\in S}(s_i^\ell-s_i^{j_i-1})>\lambda}\le e^{-\frac{\lambda^2}{4h\ell}}+e^{-\sqrt{\tau}\lambda/4} \;, %$.
\]
which completes the proof.
\end{proof}

\begin{lemma}[Completeness]
\lemlab{lem:completeness}
%\[\PPr{\calS\subset I_{\ell}}\ge1-\exp(-???).\]
With high constant probability, at the end of $\ell$ rounds of the attack,  we can find a distribution on $\mathbb{Z}^n$ such that the algorithm $\mathcal{A}$ fails with constant probability when the input is sampled from this distribution.
\end{lemma}
\begin{proof}
% \samson{Honghao, can you fill in this proof?}
Suppose that at some round $j \in [\ell]$ we have $\max\{\delta_\alpha^j, \delta_\beta^j\} = \Omega(1)$, then from Chernoff's bound we have with probability at least $0.99$, we can find this distribution $z_{I^{j - 1}}(D_\alpha)$ or $z_{I^{j - 1}}(D_\beta)$ that the algorithm $\mathcal{A}$ fails from a constant number of samples.
% the event does not happen, it means at each round, over the distribution we sampled, the $\mathsf{ALG}$ has an error probability less than some constant $c$. 

We next consider the other case where $\delta_\alpha^j, \delta_\beta^j \le c$ over the distribution $D_\alpha^n, D_\beta^n$ for all $j \in [\ell]$. Then from \lemref{lem:completeness:gap} we have 
\[
\sum_{i \in \mathcal{S}} s_i^\ell \ge 2 \ell \zeta (1 - \eta) - \lambda = \Omega(\ell)
\]
with probability at least $1  - 2\exp(-\Omega(\ell))$ by setting $\lambda = \O{\ell}$.
On the other hand, from \lemref{lem:completeness:upper} we have 
\[
\sum_{i\in\calS}s_i^\ell<\lambda+h\sigma+\frac{h}{\sqrt{\tau}} \le 3 h \sigma
\]
with probability at least $1 - 2\exp(-\Omega(\sigma))$ by setting $\lambda = h \sigma$. This is a contradiction when $\ell \ge C\cdot h \sigma$ for a sufficient large constant $C$.
\end{proof}

\subsection{Proof of Our Main Theorem}
\label{sec:main-theorem}
We are now ready to prove \thmref{thm:main-theorem}.
\begin{proofof}{\thmref{thm:main-theorem}}
    First, without loss of generality, we can assume that $n = \poly(r)$. This follows since we can always query on the first $\poly(r)$ coordinates and make the remaining $\poly(r)$ coordinates of the query vector $\bx$ to $0$ (in this case, we are attacking the first $\poly(r)$ columns of the sketching matrix $\bA$).
    
    We now prove the correctness of our attack.
    Suppose that the algorithm $\mathcal{A}$ which we attack uses the estimator $f$, and suppose we sample $\bx \sim D_{p^{t}}$ at time $t$. Next, we consider $\mathcal{A}'$ which uses the same estimator $f$, but instead takes the input $\begin{bmatrix}
        \bD \bx'\\
        \bS \bx_S
    \end{bmatrix}$
    where $\bx' \sim D_{\gamma}^{|D|}$ for a fixed $\gamma \in [\alpha, \beta]$, which is independent of input $\bx$. Since $\bD\bx$ and $\bS \bx$ are independent conditioned on $p^t$, by \lemref{lem:distribution:tvd}, we know that for each iteration $t$, the total variation distance between $\begin{bmatrix}
        \bD \bx^{(t)}_{D} \\
        \bS \bx^{(t)}_{\mathcal{S}}
    \end{bmatrix}$
    and 
    $\begin{bmatrix}
        \bD \bx' \\
        \bS \bx^{(t)}_{\mathcal{S}}
    \end{bmatrix}$
    is at most $1/\poly(n)$. Therefore, we have that $$d_{\textrm{tv}}\left(\{\mathcal{A}(\bx^{(t)})\}_{t = 1, 2, \cdots, \ell},\{\mathcal{A}'(\bx^{(t)})\}_{t = 1, 2, \cdots, \ell}\right) \leq \ell \cdot \frac{1}{\poly(n)} = \frac{1}{\poly(n)}$$

    % Let us first assume the algorithm $\mathcal{A}$ uses estimator $f'$ which only considers the input of $\bS\bx_\mathcal{S}$ and then randomly generates an input $\bx_D \sim D_{\gamma}^{|D|}$ to compute a $\bD \bx_{D}$ as another part of the input vector, where $\gamma \in [\alpha, \beta]$ is some fixed constant.
    
    Hence, it suffices for us to show that by interacting with $\mathcal{A}'$, we can find the attack distribution on which $\mathcal{A}'$ fails with high constant probability.  Note that $\mathcal{A'}$ has the property that it only uses $x_{\mathcal{S}}$ in the computation. From \lemref{soundness}, we see that with probability at least $1 - 1 / n$, we never falsely accuse any index $i \notin \mathcal{S}$; 
    % this implies that, except with probability $\frac{1}{n}$, we will never zero-out a coordinate that does not belong to the sparse part $\bS$ of $\bA$. 
    Additionally, by \lemref{lem:completeness}, we know that with high constant probability, our attack correctly identifies (some, or all) coordinates $i \in S$ and outputs a distribution on which $\mathcal{A}'$ fails. From the above discussion we can see that the algorithm $\mathcal{A}$ must also fail on this distribution with constant probability. By conditioning on these two events and taking a union bound, it follows that our attack finds some hard query distribution $\bq$ on which $\mathcal{A}$ fails with constant probability.
    
    Next, we analyze the query complexity and time complexity of our attack. In each of the $\ell$ iterations, we make $\O{1}$ queries. Thus, the total number of queries is $\O{\ell} = \O{r^8 \log^7 n} = \tO{r^{8}}$. Since we only maintain the accumulated score $s_i^t$ in each iteration $t \in [\ell]$, the total runtime of the attack is $\O{\ell n} = \poly(r)$, since it suffices to consider $n=\poly(r)$.   
\end{proofof}
% \begin{theorem}[Attack against linear sketches for $\ell_0$-estimation]
% There exists an adaptive attack that makes $\tilde O(r^{10})$ \eg{check} queries and outputs a vector $\bx$ such that $\mathcal{A}$ fails to determine whether $||x||_0 \leq (\alpha - c)n$ or $||x||_0 \geq (\beta + c)n$ with high constant probability.
% \end{theorem}
% \begin{proof}
%     Recall that from \lemref{lem:completeness:upper}, we know that the attack will correctly find and zero-out (some, or all) of the coordinates in the sparse part before outputting a hard query vector $\bx$ with probability $0.99$. Additionally, by \lemref{soundness} and \lemref{iterative:tvd}, we know that no dense coordinate will be accused throughout $\poly(r)$ iterations with probability $1- \frac{\poly(r)}{n}$. 
 
% \end{proof}

% \subsection{Moment Matching}

\section{Constructing the Hard Input Distribution}
\label{sec:distribution}
In this section, we give the construction of the hard distribution family that is used in Section~\ref{sec:attack}. 
% For the argument in Section 4, we need to argue that the dense part of the pre-processed sketching matrix $\bA'$ will not be able to distinguish between $||x||_0 \leq (\alpha + c) n$ and $||x||_0 \geq (\beta-c) n$. 
We will make use of the following lemma.
\begin{lemma}[Claim 1 of \cite{LarsenWY20}]
\lemlab{lem:cheby:at:zero}
For every $\eps>2^{-\O{R}}$, there exists a univariate polynomial $Q$ of degree at most $R-\Omega\left(\sqrt{R\log\frac{1}{\eps}}\right)$ such that
\[|Q(0)|>\eps\cdot\sum_{i=0}^R\left\lvert\binom{R}{i}\cdot Q(i)\right\rvert=\eps.\]

Furthermore, this polynomial $Q$ has the property that 
\[\sum_{i=0}^R(-1)^i\binom{R}{i}\cdot Q(i)\cdot i^t=0.\] for all non-negative integers $t \leq\O{\sqrt{R \log\frac{1}{\eps}}}$.
\end{lemma}

\begin{lemma}
\lemlab{lem:moment:match}
For any $K$, there exist constants $0\le\alpha<\beta\le 1$ such that there exists a family $\calD = \{D_p\}$ of probability distributions parameterized by $p\in[\alpha,\beta]$ with support on $\{-R,\ldots,R\}$ where $R = \O{K^2}$ such that:
\begin{enumerate}
\item
For $D_p\in\calD$, we have $D_p(0)=p$ and $D_p(1) = \Omega(1)$.
\item 
For all $p\in[\alpha,\beta]$ and for all $X\in[R]$, we have $D_p(X)=D_p(-X)$, so that $D_p$ is a symmetric distribution. 
\item 
For all $p,q\in[\alpha,\beta]$, we have $\EEx{X\sim D_p}{X^k}=\EEx{X\sim D_q}{X^k}$ for all $k\in[K]$.
\end{enumerate}
\end{lemma}
\begin{proof}
By \lemref{lem:cheby:at:zero} with $R = \Theta(K^2)$ and $\eps = 1/4$, there exists a univariate polynomial $Q$ of degree at most $R-\Omega\left(\sqrt{R }\right)$ such that
\[|Q(0)|>\frac{1}{4}\cdot\sum_{i=0}^R\left\lvert(-1)^i\binom{R}{i}\cdot Q(i)\right\rvert.\]

Moreover, for every non-negative integer $t \le K$, we have
\[
\sum_{i=0}^R(-1)^i\binom{R}{i}\cdot Q(i)\cdot i^t=0.
\]
Let $u(i)=(-1)^i\binom{R}{i}\cdot Q(i)$ for all $i\in[R]$ and let $U=\sum_{i\in[R]}|u(i)|$. 
Without loss of generality, suppose $Q(0)>0$, so that $u(0)>0$ and $u(0)>\frac{1}{4}\cdot U$. 
Moreover, since $\sum_{i\in[R]}u(i)=0$, then $u(0)\le\frac{1}{2}\cdot U$. 

We set $\alpha=\left\lvert\frac{u(0)}{2U}\right\rvert$ and $\beta=2\left\lvert\frac{u(0)}{2U}\right\rvert$. 
We first define:
\[B(i)=\begin{cases}
\left\lvert\frac{u(0)}{2U}\right\rvert,\qquad&i=0\\
\frac{1}{2}\left(\frac{1}{2}+\left\lvert\frac{u(1)}{2U}\right\rvert\right)&i=\pm1\\
\frac{1}{2}\left(\left\lvert\frac{u(i)}{2U}\right\rvert\right),\qquad&|i|\in\{2,\ldots,R\}
\end{cases}\]
Then for a fixed $p\in[\alpha,\beta]$, we define 
\[D_p(0)=B(0)+\left(\frac{p}{\alpha}-1\right)\cdot\frac{u(0)}{2U},\]
and
\[D_p(i)=B(i)+\left(\frac{p}{\alpha}-1\right)\cdot\frac{u(i)}{4U},\] for all $i$ with $|i|\in\{1,\ldots,R\}$. 

We first prove that $D_p(i)$ is a probability distribution. 
Since $\sum_{i=0}^R|u(i)|=U$, then $\sum_{i=0}^R\frac{|u(i)|}{2U}=\frac{1}{2}$, and thus 
\[\sum_{i:|i|\in\{0,1,\ldots,R\}}B(i)=\frac{1}{2} + \sum_{j=0}^R\frac{|u(j)|}{2U}=1.\]
Moreover, since $|u(i)|\le\frac{U}{2}$, then $B(i)\in[0,1]$ for all $i$ and thus $B$ is a probability distribution. 
We also have $\sum_{i = 0}^R \frac{u(i)}{U}=0$. 
Thus we have 
\begin{align*}
\sum_{i:|i|\in\{0,1,\ldots,R\}}D_p(i)&=\left(\sum_{i:|i|\in\{0,1,\ldots,R\}}B(i)\right)+\left(\sum_{i:i\in\{0,1,\ldots,R\}}\left(\frac{p}{\alpha}-1\right)\frac{u(i)}{2U}\right)\\
&=\sum_{i:|i|\in\{0,1,\ldots,R\}}B(i)=1.
\end{align*}
We also have $\sum_i\frac{|u(i)|}{2U}=\frac{1}{2}$ and thus $\frac{|u(i)|}{2U}\le\frac{1}{2}$. 
Moreover, note that for $p\in[\alpha,\beta]$ with $\alpha=\left\lvert\frac{u(0)}{2U}\right\rvert$ and $\beta=2\left\lvert\frac{u(0)}{2U}\right\rvert$, then $\left(\frac{p}{\alpha}-1\right)\in[0,1]$. 
Thus $D_p(i)\in[0,1]$ for all $i$ and so $D_p$ is a valid probability distribution. 

By construction, we have 
\begin{align*}
D_p(0)&=\left\lvert\frac{u(0)}{2U}\right\rvert+\left(\frac{p}{\alpha}-1\right)\cdot\frac{u(0)}{2U}\\
&=\left\lvert\frac{u(0)}{2U}\right\rvert+\left(\frac{2Up}{|u(0)|}-1\right)\cdot\frac{u(0)}{2U}\\
&=p,
\end{align*}
since $u(0)>0$ by assumption. 
Hence, the first part of the claim follows. 

By construction, we have $D_p$ is symmetric distribution for all $p\in[\alpha,\beta]$, which gives the second part of the claim. 

It thus remains to prove the third part of the claim. 
Let $p<q$ be fixed, for $p,q\in[\alpha,\beta]$. 
To that end, observe that $\EEx{X\sim D_p}{X^j}=\EEx{X\sim D_q}{X^j}$ if and only if $\sum_{X\in[R]}X^j\cdot(D_p(X)-D_q(X))=0$. 
Now, for each $X\in[R]$, we have $D_p(X)-D_q(X)=\frac{q-p}{\alpha}\cdot\frac{u(X)}{2U}$. 
Since $u(X)=(-1)^X\binom{R}{X}\cdot Q(X)$, then it suffices to show that $\sum_{X\in[R]}X^j\cdot(-1)^X\binom{R}{X}\cdot Q(X)=0$, which is true by \lemref{lem:cheby:at:zero}. 
Thus, the third part of the claim follows.

As an alternative view, we can first observe that since $D_p$ and $D_q$ are symmetric distributions, then their odd moments are all $0$. 
To match their even moments, we can define $\bM\in\mathbb{R}^{K\times R}$ be the following transposition of a Vandermonde matrix:
\[\bM=\begin{bmatrix}1&1&1&\ldots&1\\
1&4&9&\ldots&R^2\\
1&16&81&\ldots&R^4\\
\vdots&\vdots&\vdots&\ddots&\vdots\\
1&2^{2K}&3^{2K}&\ldots&R^{2K}
\end{bmatrix},\]
then $\EEx{X\sim D_p}{X^{2j}}$ is the $j$-th row of the matrix-vector product $\bM\bv$, where $v_i=2\cdot D_p(i)$. 
Similarly, $\EEx{X\sim D_q}{X^{2j}}$ is the $j$-th row of the matrix-vector product $\bM\bv'$, where $v'_i=2\cdot D_q(i)$ and thus $\EEx{X\sim D_p}{X^{2j}}=\EEx{X\sim D_q}{X^{2j}}$ if and only if $\bM\bv-\bM\bv'=0^K$, i.e., the all zeros vector of length $K$, so that $\bv-\bv'$ is in the kernel of $\bM$. 
Now, the $j$-th entry of $\bM\bv-\bM\bv'$ is precisely $2\sum_{X\in[R]}X^j\cdot(D_p(X)-D_q(X))=0$ and we proceed as before.
\end{proof}

\subsection{Bounding the Total Variation Distance}
Let $\bD$ denote the dense part of sketching matrix $\bA$, and let $\bx \sim D_p^n$ and $\bx' \sim D_q^n$, respectively. 
Before we proceed to prove that $\TVD(\bD \bx, \bD \bx') \leq \frac{1}{\poly(n)}$, we state the following useful lemma. 
\begin{lemma}
\lemlab{lem:prod:sum:tele}
$|\prod_{i\in[n]}(a_i+\delta_i)-\prod_{i\in[n]} a_i|\le\sum_{i\in[n]}|\delta_i|\cdot e^{\sum_{j\in[n]}|\delta_j|}$ if $|a_i+\delta_i|\le 1$ for all $i\in[n]$.
\end{lemma}
\begin{proof}
We have 
\[\left\lvert\prod_{j<i}(a_j+\delta_j)\prod_{j\ge i}a_j-\prod_{j<i+1}(a_j+\delta_j)\prod_{j\ge i+1}a_j\right\rvert=|\delta_i|\cdot\prod_{j>i}|a_j+\delta_j|\prod_{j\ge i+1}|a_j|.\]
Since $|a_j+\delta_j|\le 1$, then we have $|a_j|\le1+|\delta_j|$ by triangle inequality. 
Thus,
\begin{align*}
\left\lvert\prod_{j<i}(a_j+\delta_j)\prod_{j\ge i}a_j-\prod_{j<i+1}(a_j+\delta_j)\prod_{j\ge i+1}a_j\right\rvert&\le|\delta_i|\cdot\prod_{j\ge i+1}(1+|\delta_j|)\\
&\le|\delta_i|\prod_{j\in[n]}e^{|\delta_i|}\le|\delta_i|\cdot e^{\sum_j|\delta_j|}.
\end{align*}
Now, note that we can write
\begin{align*}
\left\lvert\prod_{i\in[n]}(a_i+\delta_i)-\prod_{i\in[n]} a_i\right\rvert=\sum_{i=1}^n\left\lvert\prod_{j<i}(a_j+\delta_j)\prod_{j\ge i}a_j-\prod_{j<i+1}(a_j+\delta_j)\prod_{j\ge i+1}a_j\right\rvert.
\end{align*}
Therefore, we have 
\begin{align*}
\left\lvert\prod_{i\in[n]}(a_i+\delta_i)-\prod_{i\in[n]} a_i\right\rvert\le\sum_{i\in[n]}|\delta_i|\cdot e^{-\sum_{j\in[n]}|\delta_j|}.
\end{align*}
\end{proof}

\begin{lemma}
\lemlab{lem:distribution:tvd}
For fixed $p$ and $p' \in [\alpha, \beta]$, let $P=D_p$ and $Q=D_{p'}$ be the pair of probability distributions defined in \lemref{lem:moment:match}. 
Let $P^n$ and $Q^n$ be the probability distributions of vectors of dimension $n$, with each entry drawn independently from $P$ and $Q$, respectively.  
Let $\bD\in \mathbb{Z}^{r\times n}$ with entries bounded in $[-\poly(n), \poly(n)]$ and 
\[
|\fracpart(\by^\top\bD)_j|^2 \le \frac{1}{s}\cdot\|\fracpart(\by^\top\bD)\|_2^2 \; . 
\]
for all $\by\in\mathbb{R}^{r}$ and $j\in[n]$. 
Let $P_{\bD}$ and $Q_{\bD}$ be the probability distributions of $\bD\bx$ and $\bD\bx'$ for $\bx\sim P^n$ and $\bx'\sim Q^n$ respectively. 
Let $K$ and $R$ be the parameter from \lemref{lem:moment:match} and $s$ be the parameter from \lemref{lem:remove:s:heavy:frac} with $s = \Omega(R^{5/2})$. 
Then the total variation distance between $P_{\bD}$ and $Q_{\bD}$ is at most $n^{\O{r}}\left(n\cdot e^{-\Omega(K)}+e^{-\Omega(K)}\right)$.
\end{lemma}
\begin{proof}
For $\bu\in[-\pi,\pi]^r$ and $\bz=\bD\bx$, we have 
\begin{align*}
\widehat{P_{\bD}}(\bu)=\EEx{\bz\sim P_{\bD}}{e^{-\langle\bu,\bz\rangle i}}=\EEx{\bz\sim P_{\bD}}{e^{-\langle\bu^\top\bD\bx\rangle i}}.
\end{align*}
We have $\bD\bx=\sum_{j\in[n]}\bD^{(j)}x_j$, where $\bD^{(j)}$ is the $j$-th column of $\bD$. 
For all $i\in[R]$, let $P_i$ be the probability that $\PPPr{X\sim P}{X=i}$. 
Since each coordinate of $\bx$ is drawn independently from $P$, then we have
\begin{align*}
\widehat{P_{\bD}}(\bu)=\prod_{j\in[n]}\Ex{e^{-\bu^\top\bD^{(j)}x_j i}}=\prod_{j\in[n]}\sum_{m\ge 0}P_m\cdot\left(\cos(\langle\bu,\bD^{(j)}\rangle m)+i\cdot\sin(\langle\bu,\bD^{(j)}\rangle m)\right).
\end{align*}
Since $P_i=P_{-i}$, then we have
\[\widehat{P_{\bD}}(\bu)=\prod_{j\in[n]}\sum_{m\ge 0}P_m\cdot\cos(\langle\bu,\bD^{(j)}\rangle m).\]
As before, we define $\fracpart(x)=x-\mathsf{int}(x)\in\left[-\frac{1}{2},\frac{1}{2}\right)$ and $\fracpart_{2\pi}(x)=2\pi\cdot\fracpart\left(\frac{x}{2\pi}\right)\in[-\pi,\pi)$, so that $\cos(m\theta)=\cos\left(m\cdot\fracpart_{2\pi}(\theta)\right)$. 
Then
\[\widehat{P_{\bD}}(\bu)=\prod_{j\in[n]}\sum_{m\ge 0}P_m\cdot\cos\left(m\cdot\fracpart_{2\pi}(\langle\bu,\bD^{(j)}\rangle)\right).\]
Rewriting $\cos(x)=1-\frac{x^2}{2!}+\frac{x^4}{4!}-\frac{x^6}{6!}+\ldots$ in its Taylor expansion, we have
\[\widehat{P_{\bD}}(\bu)=\prod_{j\in[n]}\sum_{m\ge 0}P_m\cdot\sum_{k\ge 0}\frac{\left(m\cdot\fracpart_{2\pi}(\langle\bu,\bD^{(j)}\rangle)\right)^{2k}}{(2k)!}\cdot(-1)^k.\]
Since $\cos(x)$ is well-defined, the summation is absolutely convergent, and so 
\[\widehat{P_{\bD}}(\bu)=\prod_{j\in[n]}\sum_{k\ge 0}\left(\sum_{m\ge 0}P_m\cdot m^{2k}\right)\cdot\frac{\left(\fracpart_{2\pi}(\langle\bu,\bD^{(j)}\rangle)\right)^{2k}}{(2k)!}\cdot(-1)^k.\]
Let $M_P(2k)=\left(\sum_{m\ge 0}P_m\cdot m^{2k}\right)$ be the $2k$-th moment of $P$ and $M_Q(2k)=\left(\sum_{m\ge 0}Q_m\cdot m^{2k}\right)$, so that
\[\widehat{P_{\bD}}(\bu)=\prod_{j\in[n]}\sum_{k\ge 0}M_P(2k)\cdot\frac{\left(\fracpart_{2\pi}(\langle\bu,\bD^{(j)}\rangle)\right)^{2k}}{(2k)!}\cdot(-1)^k\]
and similarly
\[\widehat{Q_{\bD}}(\bu)=\prod_{j\in[n]}\sum_{k\ge 0}M_Q(2k)\cdot\frac{\left(\fracpart_{2\pi}(\langle\bu,\bD^{(j)}\rangle)\right)^{2k}}{(2k)!}\cdot(-1)^k.\]
We claim $|\widehat{P_{\bD}}(\bu)-\widehat{Q_{\bD}}(\bu)|\le n\cdot e^{-\Omega(K)}+e^{-\Omega(K)}$ for all $\bu\in[-\pi,\pi]^n$. 
Now, for a fixed $\bu$, either there exists $j\in[n]$ such that $|\fracpart_{2\pi}(\langle\bu,\bD^{(j)}\rangle)|>\frac{1}{4K}$ or for all $j\in[n]$, we have $|\fracpart_{2\pi}(\langle\bu,\bD^{(j)}\rangle)|\le\frac{1}{4K}$. 
We analyze these cases separately. 

Suppose there exists $j\in[n]$ such that $|\fracpart_{2\pi}(\langle\bu,\bD^{(j)}\rangle)|>\frac{1}{4K}$. 
We write $\iota\left(\frac{\bu}{2\pi}\right)_j:=\fracpart_{2\pi}\langle\bu,\bD^{(j)}\rangle$. 
Then $|\iota\left(\frac{\bu}{2\pi}\right)_j|>\frac{1}{4K}$ and the definition $\fracpart_{2\pi}(x)=2\pi\cdot\fracpart\left(\frac{x}{2\pi}\right)\in[-\pi,\pi)$ implies 
\[\iota\left(\frac{\bu}{2\pi}\right)_j^2=\left\lvert\fracpart\left(\left\langle\frac{\bu}{2\pi},\bD^{(j)}\right\rangle\right)^2\right\rvert>\frac{1}{(16K^2)\cdot 2\pi}.\] 
Since we have $|\fracpart(\by^\top\bD)_j|^2<\frac{1}{s}\cdot\|\fracpart(\by^\top\bD)\|_2^2$ for all vectors $\by\in\mathbb{R}^r$, then it follows that 
\[\left\|\iota\left(\frac{\bu}{2\pi}\right)_j\right\|_2^2\ge\frac{s}{(16K^2) 2\pi} = \frac{K}{32 \pi}.\]
by setting $s = \O{K^3}$. From before, we have
\begin{align*}
|\widehat{P_{\bD}}(\bu)|&=\left\lvert\prod_{j\in[n]}\sum_{m\ge 0}P_m\cdot\cos\left(m\cdot\fracpart_{2\pi}(\langle\bu,\bD^{(j)}\rangle)\right)\right\rvert\\
&=\prod_{j\in[n]}\left\lvert\sum_{m\ge 0}P_m\cdot\cos\left(m\cdot\iota\left(\frac{\bu}{2\pi}\right)_j\cdot 2\pi\right)\right\rvert. 
\end{align*}
Since we have $P_1=\Omega(1)$, then
\begin{align*}
|\widehat{P_{\bD}}(\bu)|&\le\prod_{j\in[n]}\left\lvert1-P_1\left(1-\cos\left(m\cdot\iota\left(\frac{\bu}{2\pi}\right)_j\cdot 2\pi\right)\right)\right\rvert\\
&\le\prod_{j\in[n]} e^{-\Omega\left(\left(\iota\left(\frac{\bu}{2\pi}\right)_j\right)^2\right)},
\end{align*}
where the last inequality holds by the Taylor expansion $\cos(x)=1-\frac{x^2}{2!}+\frac{x^4}{4!}-\frac{x^6}{6!}+\ldots$ and the inequality $1-x\le e^{-x}$. 
We thus have
\begin{align*}
|\widehat{P_{\bD}}(\bu)|&\le e^{-\Omega\left(\left\|\iota\left(\frac{\bu}{2\pi}\right)\right\|_2^2\right)}\\
&\le e^{-\Omega(K)},
\end{align*}
and similarly $|\widehat{Q_{\bD}}(\bu)|\le e^{-\Omega(K)}$. 
Thus in this case, $|\widehat{P_{\bD}}(\bu)-\widehat{Q_{\bD}}(\bu)|\le e^{-\Omega(K)}$, by triangle inequality. 

In the other case, we have that for all $j\in[n]$, $|\fracpart_{2\pi}(\langle\bu,\bD^{(j)}\rangle)|\le\frac{1}{4K}$. 
From before, we have
\[\widehat{P_{\bD}}(\bu)=\prod_{j\in[n]}\sum_{k\ge 0}M_P(2k)\cdot\frac{\left(\fracpart_{2\pi}(\langle\bu,\bD^{(j)}\rangle)\right)^{2k}}{(2k)!}\cdot(-1)^k.\]
At this point, we recall that $R = \O{K^2}$ by \lemref{lem:moment:match}. So, using $R= \O{K^2}$ and the fact that for all $j\in[n]$, $|\fracpart_{2\pi}(\langle\bu,\bD^{(j)}\rangle)|\le\frac{1}{4K}$ (as well as Stirling's approximation), we can upper bound the higher moments as follows:
\begin{align*}
\left\lvert\sum_{k\ge K/2}M_P(2k)\cdot\frac{\left(\fracpart_{2\pi}(\langle\bu,\bD^{(j)}\rangle)\right)^{2k}}{(2k)!}\cdot(-1)^k\right\rvert\le\sum_{k>K/2}R^{2k}\cdot\frac{1}{(2k)!}\cdot \left(\frac{1}{16K^2}\right)^k \\ \le  \frac{K^{4K}}{(2K)^{2K}/e^{2K}\cdot \sqrt{4\pi K} \cdot (16)^K} \cdot \frac{1}{K^{2K}} \leq e^{-\Omega(K)}
\end{align*}
We now apply \lemref{lem:prod:sum:tele} with $a_j=\sum_{k\le K/2}M_P(2k)$ and $\delta_j=\sum_{k>K/2}M_P(2k)$ so that
\[\left\lvert\widehat{P_{\bD}}(\bu)-\prod_{j\in[n]}\sum_{k\le K/2}M_P(2k)\cdot\frac{\left(\fracpart_{2\pi}(\langle\bu,\bD^{(j)}\rangle)\right)^{2k}}{(2k)!}\cdot(-1)^k\right\rvert\le n\cdot e^{-\Omega(K)}.\]
Similarly, we have
\[\left\lvert\widehat{Q_{\bD}}(\bu)-\prod_{j\in[n]}\sum_{k\le K/2}M_Q(2k)\cdot\frac{\left(\fracpart_{2\pi}(\langle\bu,\bD^{(j)}\rangle)\right)^{2k}}{(2k)!}\cdot(-1)^k\right\rvert\le n\cdot e^{-\Omega(K)}.\]
Moreover, we have $M_Q(2k)=M_Q(2k)$ for $k\le K/2$ and thus by triangle inequality, we have $|\widehat{P_{\bD}}(\bu)-\widehat{Q_{\bD}}(\bu)|\le n\cdot e^{-\Omega(K)}$. 

Thus, combining both cases, we have $|\widehat{P_{\bD}}(\bu)-\widehat{Q_{\bD}}(\bu)|\le n\cdot e^{-\Omega(K)}+e^{-\Omega(K)}$ for all $\bu\in[-\pi,\pi]^n$, as desired. Now, we have
\begin{align*}
|P_{\bD}(\bx)-Q_{\bD}(\bx)|&=\left\lvert\frac{1}{(2\pi)^r}\int_{[-\pi,\pi)^r}e^{i\langle\bu,\bx\rangle}\left(\widehat{P_{\bD}}(\bu)-\widehat{Q_{\bD}}(\bu)\right)\,d\bu\right\rvert\\
&\le n\cdot e^{-\Omega(K)}+e^{-\Omega(K)}.
\end{align*}
Finally, we observe that since $\bD\in \mathbb{Z}^{r\times n}$ with entries bounded in $[-\poly(n), \poly(n)]$, then $P_{\bD}(\bx)$ and $Q_{\bD}(\bx)$ only have support on a set of size $n^{\O{r}}$. 
Thus, we have that $$d_{\mathrm{tv}}(P_{\bD}(\bx), Q_{\bD}(\bx)) \leq n^{\O{r}}\left(n\cdot e^{-\Omega(K)}+e^{-\Omega(K)}\right)$$
\end{proof}

At this point, we note that $s = \O{K^3}$ in the proof of \lemref{lem:distribution:tvd}. So, by setting $K = r \log n$, we see that $s = \O{(r\log n)^3}$. For this choice of parameters $s, K$, we get that $\TVD(P_{\bD}(\bx), Q_{\bD}(\bx)) \leq \frac{1}{\poly(n)}$, as desired.

\section{Attack against Linear Sketches over Finite Fields}
In this section, we present our attack against linear sketches for $\ell_0$-estimation in the case that the sketching matrix $\bA \in \mathbb{F}_p^{r \times n}$ and inputs $\bx \in \mathbb{F}_p^n$ come from a finite field for some prime $p$. Formally, we have the following theorem.
\begin{theorem}\thmlab{finitefield}
\thmlab{thm:fp}
There exists an adaptive attack that makes $\tO{r^3}$ queries and with high constant probability outputs a distribution $D$ over $\mathbb{Z}^{n}$ such that when $\bx \sim D$, $\calA$ fails to distinguish between $\|x\|_0 \leq 1.1 n$ and $\|x\|_0 \geq 1.9 n$ with constant probability.
\end{theorem}

\begin{algorithm}
\caption{Attack on $L_0$ algorithms that use a sketching matrix over $\mathbb{F}_p^{r\times n}$}
\alglab{alg:lzero:attack}
\begin{algorithmic}[1]
\Require{Algorithm $\calA$ that decides whether input vector $\bx$ satisfies $\|\bx\|_0\le 1.1 r$ or $\|\bx\|_0\ge 1.9r$, using a sketch matrix $\bA\in\mathbb{F}_p^{r\times n}$}
\Ensure{A query distribution on which $\calA$ does not succeed with constant probability}
\State{$T\gets\emptyset$}
\While{$|T|<r$}
\State{Randomly choose $R\subseteq([n]\setminus T)$ of size $2r$}
\State{Let $\bx^{(1)}\in\mathbb{F}_p^n$ be a random vector with support only on $T$}
\State{Let $\bx^{(2)} \in \mathbb{F}_p^n$ be a random vector with support only on $T \cup R$} 
\If {$\mathcal{A}$ fails on $\bx^{(1)}$ or $\bx^{(2)}$}
\State{Return this distribution $\bx^{(i)}$}
\EndIf
\For{$\ell=1$ to $\ell=5+\log\log r+\log r$}
\Comment{$2^\ell$ indices of TVD $\O{\frac{1}{2^\ell}}$}
% \If{$\ConfirmIndex(T,\ell)$}
%\If{$\TVD(\calD_1,\calD_2)=\Omega\left(\frac{1}{2^{\ell}}\right)$}
\If{$\FindColumn(T,R,\ell)$ outputs a column $j$}
\State{$T\gets T\cup\{j\}$}
% \State{Restart the outer loop at $\ell=1$}
% \EndIf
\EndIf
\EndFor
\EndWhile
\State{Let $\bx^{(1)}\in\mathbb{F}_p^n$ be a random vector with support only on $T$}
\State{Let $\bx^{(2)}$ be a random vector from $\mathbb{F}_p^n$} 
%that is on the support of $T$ and has support size $r$ outside of $T$} 
% \State{With probability $\frac{1}{2}$, let $\bx=\bx^{(1)}$, otherwise let $\bx=\bx^{(1)}+\bx^{(2)}$}
\State{Return one of $\bx^{(1)}$ and $\bx^{(2)}$}
\end{algorithmic}
\end{algorithm}

\begin{algorithm}[!htb]
\caption{$\FindColumn(T,R,\ell)$}
\alglab{Find a column $j$ linearly independent of $T$}
\begin{algorithmic}[1]
\Require{Set $T$, Set $R, \ell\in[5 + \log\log r +\log r]$}
\Ensure{A column $j$ that is linear independent to $T$}
\State{Let $R^i$ denote the first columns of $R$}
\For{$m_3=\O{\frac{r}{2^{\ell}}\log r}$ times}
% \State{Randomly choose $j\in([n]\setminus T)$}
\State{Randomly choose $i \in [2r]$}
\For{$m_4=\O{2^{2\ell}\log r}$ times}
% \State{Randomly choose $R\subseteq([n]\setminus(T\cup\{j\}))$ of size $i$}
\State{Randomly generate $\bv^{(3)}\in\mathbb{F}_p^n$ with support only on $R^{i}\cup T$}
\State{Randomly generate $\bv^{(4)}\in\mathbb{F}_p^n$ with support only on 
$R^{i + 1}\cup T$}
\State{Query $\calA$ on $\bv^{(3)}$ and $\bv^{(4)}$}
\EndFor
\State{Let $\calD_3$ and $\calD_4$ be the output distributions of $\{\bv^{(3)}\}$ and $\{\bv^{(4)}\}$}
\If{$\TVD(\calD_3,\calD_4)\ge\frac{1}{2^{\ell+3}\log r}$}
\State{\Return $j$}
%\State{$T\gets T\cup\{j\}$}
%\State{Restart the outer loop at $\ell=1$}
\EndIf
\EndFor
\State{\Return FAIL}
\end{algorithmic}
\end{algorithm}

The full description of our algorithm is given in \algref{alg:lzero:attack}. The basic idea of our attack is due to the following observation: let $T$ and $R$ be two subsets of columns in $\bA$ such that $T$ and $R$ have the same column span. Then $\TVD(\bA\bx^{(1)}, \bA\bx^{(2)})=0$, 
where $\bx^{(1)}\in\mathbb{F}_p^n$ and $\bx^{(2)}\in\mathbb{F}_p^n$ are uniformly random vectors with support on $T$ and $R$, respectively (Corollary~\ref{lem:same-tvd}). Thus, if we can find a column-independent set $T$ with $r$ columns, the algorithm $\mathcal{A}$ must fail on one of the following two cases where $\bx$ is a random vector that is on the support $T$ or a random vector over $\mathbb{F}_p^n$, as they correspond to the different outputs of $\mathcal{A}$.
Therefore, the remaining task is to devise a strategy to find the column independent set $T$. 

\begin{lemma}
\lemlab{lem:depend:tvd:zero}
Let $T$ be a subset of columns in $\bA$ and suppose column $j$ is linearly dependent with the columns in $T$. 
Then $\TVD(\bA\bv^{(1)}, \bA\bv^{(2)})=0$, i.e., the distributions of the sketch on $\bv^{(1)}$ and $\bv^{(2)}$ are identical. 
Here $\bx^{(1)}\in\mathbb{F}_p^n$ is random vector with support on $T$ and $\bx^{(2)}\in\mathbb{F}_p^n$ be a random vector with support on $T \cup \{j\}$.
\end{lemma}
\begin{proof}
From the condition, we have that there exist $\alpha_1,\ldots,\alpha_{|T|}\in\mathbb{F}_p$ such that
\[\alpha_1\bA^{(T_1)}+\ldots+\alpha_{|T|}\bA^{(T_{|T|})}=\bA^{(j)}.\]
% Hence, we have the distribution of 

Thus there exists a one-to-one correspondence for the setting where the coordinate of $\bv^{(1)}$ corresponding to the $i$-th index of $T$ is $\beta_i\in\mathbb{F}_p$ and the setting where the coordinate of $\bv^{(2)}$ corresponding to the $i$-th index of $T$ is $\beta_i$, i.e., the coordinate of $\bv^{(1)}$ corresponding to the $i$-th index of $T$ is $\beta_i-\alpha_i$. 
Thus, the output distributions of the sketch on $\bv^{(1)}$ and $\bv^{(2)}$ are identical, i.e., $\TVD(\bA\bv^{(1)}, \bA\bv^{(2)})=0$.  
\end{proof}

\begin{corollary}
    \label{lem:same-tvd}
    Let $T$ and $R$ be two subsets of columns in $\bA$ and suppose that they have the same column span.
Then $\TVD(\bA\bx^{(1)}, \bA\bx^{(2)})=0$, 
where $\bx^{(1)}\in\mathbb{F}_p^n$ is random vector with support on $T$ and $\bx^{(2)}\in\mathbb{F}_p^n$ be a random vector with support on $R$.
\end{corollary}

% The procedure we search for this column-independent set is as follows: 
We next give some high-level intuition of our procedure that searches for this column-independent set:
suppose $T$ is the current set of linear columns found, then we randomly sample $2r$ columns in $[n] \setminus T$, and let $R$ denote the set of these new columns. Then from the correctness guarantee of the algorithm $\mathcal{A}$ we have that $\TVD(\mathcal{A}(\bx^{(1)}), \mathcal{A}(\bx^{(2)})) \ge 1/3$ (as otherwise we find the distribution on which $\mathcal{A}$ fails immediately), 
where $\bx^{(1)}$ is a random vector with support on $T$ and $\bx^{(2)}$ is a random vector with support on $T + R$. Next let $R^i$ denote the first $i$ columns in $R$ and  $\mu_i$ denote the distribution of $\mathcal{A}(\bx^{(i)})$ where  $\bx^{(i)}$ is the random vector in the support of $T \cup R^i$. From the triangle inequality we have 
\begin{equation}
\sum_i d_{\mathrm{tv}}(\mu_i, \mu_{i + 1}) \ge  d_{\mathrm{tv}}(\mu_0, \mu_{2r}) \ge \frac{1}{3}.
\end{equation}
One natural way at this point is from the above, we have there must exist $j$ such that $ d_{\mathrm{tv}}(\mu_{j - 1}, \mu_{j}) \ge \Omega(1/r)$, and then such $j$ should be a column that is linearly independent to the columns in $T$, as otherwise the total variation distance should be $0$. Hence, we can enumerate all $i \in [2r]$ to find such column $j$ (from the results in statistical testing, we can distinguish whether two binary distributions have $0$ distance or have total variation distance larger than $1/r$ using $\widetilde{O}(r^2)$ samples with error probability at most $1/\poly(r)$ (\lemref{lem:testing})).  However, such a way might not be optimal, as in the worst case we need to search every $i \in [2r]$. To get a better $r$ dependence, we consider the following level-set argument: define the level set $I_0 = [\frac{1}{\log r},1)$ and $I_\ell = \left[\frac{1}{2^{\ell+3}\log r},\frac{1}{2^{\ell+2}\log r}\right)$, then since 
$\sum_i d_{\mathrm{tv}}(\mu_i, \mu_{i + 1}) \ge \frac{1}{3}$,
there exists $\ell\in[5+\log\log r+\log r]$ for which there exist at least $2^{\ell-1}$ indices $i$ such that $\TVD(\mu_i,\mu_{i+1})\in I_{\ell - 1}$ (\lemref{lem:ell:exists}). Hence, we can guess the value of $\ell$, and for each value of $\ell$, we use a proper sampling rate to sample the indices in $[2r]$. Note that since the range of the total variation distance is different for each $\ell$, we can use different number of samples (which depends on $\ell$) to do the distribution testing. This results in a better $r^3$ dependence.
    
\begin{lemma}
\lemlab{lem:ell:exists}
Suppose that 
\[
\sum_{i = 0}^{2r - 1} d_{\mathrm{tv}}(\mu_i, \mu_{i + 1}) \ge  \frac{1}{3}.
\]
Define the level set $I_0 = [\frac{1}{\log r},1)$ and $I_\ell = \left[\frac{1}{2^{\ell+3}\log r},\frac{1}{2^{\ell+2}\log r}\right)$.
There exists $\ell\in[5+\log\log r+\log r]$ for which there exist at least $2^{\ell-1}$ indices $i$ such that $\TVD(\mu_i,\mu_{i+1})\in I_{\ell - 1}$.
\end{lemma}
\begin{proof}
Suppose by way of contradiction that for all $\ell\in[5+\log\log r+\log r]$, there exists fewer than $2^{\ell-1}$ indices $i$ such that $\TVD(\mu_i,\mu_{i+1})\in I_{\ell-1}$. 
Let $N_\ell$ be the number of indices $i$ such that $\TVD(\mu_i,\mu_{i+1})\in I_{\ell - 1}$. 
Then we have
\[\sum_{i = 0}^{2r - 1} d_{\mathrm{tv}}(\mu_i, \mu_{i + 1})\le\sum_{\ell=1}^{5+\log\log r+\log r} \frac{N_\ell}{2^{\ell+1}\log r} + 2r \cdot \frac{1}{32r} < \frac{1}{4} \]
\end{proof}

%exists

Before proving our main theorem. We need the following result in the discrete distribution testing.

\begin{lemma}[\cite{CDVV14}]
\lemlab{lem:testing}
    Suppose that $p$ and $q$ are two distributions on $[n]$ There is an algorithm that uses $\O{\max\{n^{2/3}/\eps^{4/3}, n^{1/2}/\eps^{2}\}}$ samples to distinguish whether $p = q$ or $d_{\mathrm{tv}}(p, q) \ge \eps$ with probability at least $2/3$.
\end{lemma}

Note that the distributions we test is binary as the algorithm $\mathcal{A}$ only output $0$ or $1$. And to boost the error probability to $\delta$, we can run $\log(1/\delta)$ independent copies and then take the majority.

We are now ready to prove our \thmref{thm:fp}. 
\begin{proofof}{\thmref{thm:fp}}
Consider \algref{alg:lzero:attack}. With probability at least $1 - 1/\poly(r)$, all of the distribution testing subroutines succeeded, this is because we make an extra of $\O{\log r}$ factor in the number of samples for each testing procedure and take a union bound. Condition on this event, we only need to show in each iteration, with probability at least $1 - 1/\poly(r)$ we can find a new column $j$ that is linearly independent to $T$.
    
Consider a fixed iteration and let $\bx^{(1)}$ is a random vector with support on $T$ and $\bx^{(2)}$ is a random vector with support on $T + R$.
We first consider the case where $\TVD(\mathcal{A}(\bx^{(1)}), \mathcal{A}(\bx^{(2)})) \le 1/3$, 
then from the guarantee of the algorithm $\mathcal{A}$, $\mathcal{A}$ must fail on one of the distributions.  
    
    We next consider the other case $\TVD(\mathcal{A}(\bx^{(1)}), \mathcal{A}(\bx^{(2)})) \ge 1/3$. First, if during the process, $\FindColumn$ successfully finds a column $j$, since we assume the correctness of the property testing subroutines, this means column $j$ must be linearly independent to $T$ (as otherwise the total variation distance is $0$). One the other hand, from \lemref{lem:ell:exists}, we know that there exists there exists $\ell\in[5+\log\log r+\log r]$ for which there exist at least $2^{\ell-1}$ indices $i$ such that $\TVD(\mu_i,\mu_{i+1})\in I_{\ell - 1}$. Since we sample $\O{\frac{r}{2^{\ell}}\log r}$ index $i$ in this range, with probability at least $1 - 1/\poly(r)$, we can find such a $j$ that $\TVD(\mu_j,\mu_{j+1})\in I_{\ell - 1}$.

    Now, assume that we have found such a column-independent set $T$ with $r$ columns. Let  $\bx^{(1)}\in\mathbb{F}_p^n$ is random vector with support on $T$ and $\bx^{(2)}\in\mathbb{F}_p^n$ be a random vector on $\mathbb{F}_p^n$ . Recall that $\bA \in \mathbb{F}_p^{r\times n}$, this means that we have $d_{\mathrm{tv}} (\bA \bx^{(1)}, \bA \bx^{(2)}) = 0$, which means that the algorithm $\mathcal{A}$ must fail on one of the distributions.

    Finally, we analyze the query complexity. in each step of the finding of the $r$ columns in $T$, we make $\log r + \log\log r + 5$ guess about the value of $\ell$ and in each guess we sample $\O{\frac{r}{2^\ell} \log r}$ column $j$ and in each sample we make $\O{2^{2\ell}\log r}$ samples of the two distributions, then it follows that the overall query complexity is
    \[
    r \cdot \left(\sum_{\ell = 1}^{(\log r + \log\log r + 5)} \frac{r}{2^\ell} \log r \cdot 2^{2\ell}\log r \right) = r^3 \cdot \polylog(r) \;. \qedhere
    \]
\end{proofof}

\section{Attack against Real-Valued Linear Sketches}

In this section, we consider the case where the sketching matrix $\bA \in \mathbb{R}^{r \times n}$ has all subdeterminants at least $\frac{1}{\poly(r)}$ (note that the known sketches have this property).
%$\bA$ is a real matrix over $\mathbb{R}^{r \times n}$. 
Formally, we prove the following theorem.

\begin{restatable}{theorem}{thmreal}
\thmlab{thm:real}
Suppose that $\bA$ with the estimator $f$ solves the $(\alpha + c, \beta - c)$- $\ell_0$ gap norm problem with some constants $\alpha, \beta$, and $c$, where $\bA \in \mathbb{R}^{r \times n}$ is the sketching matrix and has all nonzero subdeterminants at least $\frac{1}{\poly(r)}$, and $f: \mathbb{R}^{r \times n} \rightarrow \{-1, +1\}$ is any estimator used by $\mathcal{A}$, and $\mathcal{A}$ returns $f(\bA, \bA \bx)$ for each query $\bx$. 
    
Then, there exists a randomized algorithm, which after making an adaptive sequence of queries to $\mathcal{A}$, with high constant probability can generate a distribution $D$ on $\mathbb{R}^n$ such that $\mathcal{A}$ fails on $D$ with constant probability. Moreover, this adaptive attack algorithm makes at most $\poly(r)$ queries and runs in $\poly(r)$ time. %  \polylog(n)    
\end{restatable}

We follow a similar procedure as we did for $\bx \in \mathbb{Z}^{n}$, where the strategy is to design queries to learn the significant columns of the sketching matrix $\bA$. However, since the sketching matrix $\bA$ is real-valued, we may need to redefine the significance of columns and re-design the hard input distribution family for the insignificant coordinates. Specifically, we consider the following condition for the significance of column $i$:
 \[
 % \begin{equation}
 % \label{eq:real1}
    \exists \by^\top\in \mathbb{R}^r, (\by^\top \bA)_i^2\geq \frac{1}{s}\cdot \|\by^\top \bA\|_2^2 \; .
% \end{equation}
\]

Next, we argue that we can iteratively remove a (small) number of columns of a matrix $\bA\in\mathbb{R}^{r\times n}$ such that the resulting matrix $\bA'$ has leverage scores at most $\frac{1}{s}$ (note that since $\bA$ and $\bx$ are real-valued matrix and vector now, the previous information theoretic argument in Section~\ref{sec:pre-processing} no longer works).
Since the sum of the leverage scores is at most $r$, we would like to argue that we can just remove $rs$ columns. 
However, this may not be true, since the leverage scores of some columns may increase when we zero-out other columns during the pre-processing. Thus, we require a more involved volume argument to bound the total number of added rows $e_i$, which has previously been used to bound the sum of online leverage scores~\cite{CohenMP20,BravermanDMMUWZ20}. 
\begin{lemma}[Matrix determinant lemma]
\lemlab{lem:matrix:det}
For any vector $\bu\in\mathbb{R}^d$ and matrix $\bM\in\mathbb{R}^{d\times d}$, we have 
\[\det(\bM+\bu\bu^\top)=\det(\bM)\cdot(1+\bu\bM^{-1}\bu).\]
\end{lemma}
Now, we show that for matrices with \textit{bounded entries} and \textit{bounded subdeterminants}, we can only zero-out columns with high leverage scores for a fixed number of times before the remaining columns have bounded leverage score. Among this class of matrices is the class of integer matrices with bounded entries. 
We remark that if each entry in a general matrix is represented using $b$ bits, then by rescaling, this translates to an integer matrix whose entries are bounded by at most $2^b$ in magnitude. 
%\david{we should note that the non-adaptive sketches for l0 work also over the reals and the entries are all integers in {0, 1, 2, 4, 8, ..., n}, where n is the number of coordinates}
%\david{you could add that for l2 one could also use AMS with entries in {-1,1} and it works over the reals}

We further remark that although the following statement for matrices with subdeterminant at least $\frac{1}{\poly(r)}$, the statement easily extends to matrices with subdeterminants at least $\kappa$ by removing $\O{rs\log(\kappa nr)}$ columns, e.g., matrices with subdeterminants at least $\frac{1}{n^{\poly(r)}}$ would require $\poly(r)\cdot\log n$ columns to be removed. 
\begin{lemma}
Let $\bA\in\mathbb{R}^{r\times n}$ be a matrix with nonzero entries bounded by $\poly(r)$ and all subdeterminants either zero or at least $\frac{1}{\poly(r)}$. 
Let $s\ge 1$ be a given parameter. 
Then there exists a pre-processing procedure to $\bA$ that produces a matrix $\bA'\in\mathbb{Z}^{r\times n}$ that zeros out at most $\O{r^2s\log(nr)}$ columns of $\bA$ such that the leverage score of all columns of $\bA'$ is at most $\frac{1}{s}$. 
\end{lemma}
\begin{proof}
Let $\bS = \bA\bA^\top\in\mathbb{R}^{r\times r}$. 
By the matrix determinant lemma, c.f., \lemref{lem:matrix:det}, we have for any vector $\bu\in\mathbb{R}^r$, $\det(\bS+\bu\bu^\top)=\det(S)\cdot(1+\bu^\top\bS^{-1}\bu)$. 
Suppose a column $\bA_i$ is removed from the sketching matrix $\bA$, so that $\bS$ decreases by $\bA_i\bA_i^\top$. 
By the matrix determinant lemma, c.f., \lemref{lem:matrix:det}, we have $\det(\bS-\bA_i\bA_i^\top) = \det(S) (1-(\bA_i^\top\bS^{-1}\bA_i))$. 
Note by the definition of leverage score, $\bA_i^\top\bS^{-1}\bA_i$ is the $i$-th leverage score $\ell_i$ of the current $\bS$. 
Hence, we have $\det(\bS-\bA_i\bA_i^\top)=\det(S)(1-\ell_i)$. 
Observe that if $\ell_i=1$, then the rank of $\bS$ decreases, and the analysis can be restarted with a new linearly independent subset of columns of the matrix $\bA$ at that time. 
Thus we can have $\ell_i=1$ at most $r$ times and for the remainder of the analysis, we shall consider the number of columns that must be removed while not decreasing the rank of $\bA$. 

Now in the case that no columns have leverage score $1$, we seek to remove columns with leverage score $\ell_i>\frac{1}{s}$. 
In this case, we have $|\det(\bS-\bA_i\bA_i^\top)|\le|\det(\bS)|\cdot\left(1-\frac{1}{s}\right)$. 
%We can also assume without loss of generality that the rows of $\bA$ are orthonormal, since the sketching algorithm can always apply the appropriate post-processing to recover the image of a general sketching matrix from the image of the corresponding sketching matrix after its rows have been orthonormalized. 
On the other hand, we have that $|\det(\bS)|\le\|\bS\|_F^r\le(n\cdot\poly(r))^r$, since $\bS=\bA^\top\bA$ so that $\|\bS\|_F\le\|\bA\|_F^2\le n\cdot\poly(r)$ since each of the entries of $\bA$ have magnitude at most $\poly(r)$. 
Hence after $\O{rs\log(nr)}$ iterations of removing columns with leverage score at least $\frac{1}{s}$, we have $|\det(\bS-\bA_i\bA_i^\top)|<\frac{1}{\poly(r)}$, which contradicts $|\det(\bS-\bA_i\bA_i^\top)|\ge\frac{1}{\poly(r)}$, given the assumption that any subdeterminant of $\bA$ has value at least $\frac{1}{\poly(r)}$. 
Thus, we remove $\O{rs\log(nr)}$ columns for a given rank, and have at most $r$ changes to the rank of the matrix. 
Therefore, at most $\O{r^2s\log(nr)}$ columns are removed in total before no remaining columns have leverage score at last $\frac{1}{s}$.  
\end{proof}

We next consider the construction of the hard distribution for the insignificant coordinates. Let $D=\calN(0,1)$ and let $D_p$ for a constant $p\in(0,1)$ be the distribution such that for $x\sim D_p$ satisfies $\PPr{x=0}=1-p$. 
Otherwise, with probability $p$, $x\sim\calN\left(0,\frac{1}{p}\right)$. 
Note that we can also write $x\sim D_p$ by $x=\frac{1}{\sqrt{p}}\cdot\Bern\left(p\right)\cdot\calN(0,1)$, where $\Bern\left(p\right)$ denotes a Bernoulli random variable with parameter $p$, i.e., $1$ with probability $p$ and $0$ with probability $1-p$, 
Thus, we have
\[\EEx{x\sim\calD_1}{x}=\EEx{x\sim\calD_2}{x}=0,\qquad\EEx{x\sim\calD_1}{x^2}=\EEx{x\sim\calD_2}{x^2}=1.\]

We next turn to bound the total variation distance $\TVD(\bA \bx^{(1)}, \bA \bx^{(2)})$ where $\bx^{(1)} \sim D_p $ and $\bx^{(2)} \sim D_q$ for $p$ and $q$ randomly sampled in $(\alpha, 1)$ for some small constant $\alpha$. 
We first recall the following statement of Azuma's inequality:
%are sampled from the distributions in the distribution family  , 
\begin{theorem}[Azuma's inequality]
\thmlab{thm:azuma:ineq}
Let $Z_1,\ldots,Z_n$ be mean-zero random variables and $\beta_1,\ldots,\beta_n$ be upper bounds such that for all $i\in[n]$, $|Z_i|\le\beta_i$. 
Then 
\[\PPr{\left\lvert\sum_{i=1}^nZ_i\right\rvert>t}\le\exp\left(-\frac{t^2}{2\sum_{i\in[n]}\beta_i^2}\right).\]
\end{theorem}
We next show that a random matrix $\bB$ formed by rescaling columns sampled from a matrix $\bA$ is a good subspace embedding of $\bA$. 
\begin{lemma}
\lemlab{lem:spectral:preserve}
Let $\gamma\ge 1$ be a fixed constant. 
Let $\bA\in\mathbb{R}^{r\times n}$ and $s=\Theta\left(\frac{\gamma^2}{p^2}\cdot r^4\log r\right)$ be fixed so that no column of $\bA$ has leverage score more than $\frac{1}{s}$. 
Let $\bB\in\mathbb{R}^{r\times n}$ be a random matrix formed by sampling and scaling by $\frac{1}{p}$ each column of $\bA$ with probability $p$, otherwise zeroing out the column entirely. 
Then with high probability, we have that simultaneously for all $\bx\in\mathbb{R}^r$
\[\left(1-\frac{1}{\gamma r}\right)\cdot\|\bx^\top\bB\|_2^2\le\|\bx^\top\bA\|_2^2\le\left(1+\frac{1}{\gamma r}\right)\cdot\|\bx^\top\bB\|_2^2.\]
\end{lemma}
\begin{proof}
Let $\bx\in\mathbb{R}^r$ be any fixed vector and let $\by=\bA^\top\bx\in\mathbb{R}^n$. 
If $\by$ is the all zeros vector, then all columns of $\bA$ have dot product $0$ with $\bx$ and so $\by=\bx^\top\bB$. 
Thus it suffices to consider the case where $\by$ is nonzero, in which case we can suppose $\by$ has unit $L_2$ norm without loss of generality. 
Now for $i\in[n]$, let $Z_i=\frac{1}{p}\cdot y_i^2\cdot X_i-y_i^2$, where $X_i\sim\Bern\left(p\right)$. 
Hence, we have $\Ex{Z_i}=0$ and $|Z_i|\le\left(\frac{1}{p}-1\right)y_i^2\le\left(\frac{1}{p}-1\right)\ell_i$, where $\ell_i$ is the leverage score of column $i$, since $\ell_i=\max_{\bv\in\mathbb{R}^r}\frac{\langle\bv,\ba_i\rangle^2}{\|\bA^\top\bv\|_2^2}$. 
Thus for all $i\in[n]$, we can set $\beta_i=\frac{\ell_i}{p}$, so that $\sum_{i\in[n]}\beta_i\le\sum_{i\in[n]}\frac{\ell_i}{p}\le\frac{r}{p}$. 
Moreover, we have that for all $i\in[r]$, $\ell_i\le\frac{1}{s}$ by our pre-processing, and thus $\beta_i\le\frac{1}{s}$. 
We also have $\sum_{i=1}^ny_i^2=\|\by\|_2^2=1$. 
Hence for $s=\Theta\left(\frac{\gamma^2}{p^2}\cdot r^4\log r\right)$, Azuma's inequality, c.f., \thmref{thm:azuma:ineq},  implies
\begin{align*}
    \PPr{1-\sum_{i\in[n]}2\cdot y_i^2\cdot X_i>\frac{1}{\gamma r}} \le\exp\left(-\frac{p^2}{\gamma^2r^2\cdot 2\cdot\frac{r}{s}}\right) =\exp(-\Theta(\gamma^2 r\log r)),
\end{align*}
for $s=\Theta\left(\frac{\gamma^2}{p^2}\cdot r^4\log r\right)$. 

Now we take a $\frac{1}{\gamma r}$-net $\calN$ of the unit vectors $\by$ in the row-span of $\bA$. 
We have $|\calN|\le(\gamma r)^{\O{r}}$ and thus by a union bound over all $\calN$, we have that all net points $\by'\in\calN$ have their length preserved up to $1\pm\frac{1}{\gamma r}$. 
Finally to show that correctness over the net $\calN$ implies correctness everywhere, we can view our estimation procedure as generating a diagonal sampling matrix $\bD$ with $\sqrt{1/p}$ on the diagonal entries corresponding to the columns sampled into $\bB$, and $0$ otherwise. 
Thus for an arbitrary unit vector $\by$ in the row-span of $\bA$, let $\by'$ be the vector of $\calN$ closest to $\by$. 
Then by triangle inequality, we have
\begin{align*}
\|\bD\by\|_2&\le \|\bD\by'\|_2 + \|\bD(\by-\by')\|_2 \le 1+\frac{1}{\gamma r}+\sqrt{2}\cdot\|\by-\by\|_2\le1+\O{\frac{1}{\gamma r}},\\
\|\bD\by\|_2&\ge \|\bD\by'\|_2 - \|\bD(\by-\by')\|_2 \ge 1-\frac{1}{\gamma r}-\sqrt{2}\cdot\|\by-\by\|_2\ge1-\O{\frac{1}{\gamma r}}.
\end{align*}
Since $\bD\by=\bx^\top\bB$, then the desired claim follows. 
\end{proof}
Recall the following definition of KL divergence: 
\begin{definition}[Kullback-Leibler Divergence]
For continuous distributions $P$ and $Q$ of a random variable with probability densities $p$ and $q$ on support $\Omega$, the KL divergence is
\[\KLD(P||Q)=\int_{x\in\Omega}p(x)\log\frac{p(x)}{q(x)}\,dx.\]
\end{definition}
The following statement about the KL divergence of a multivariate Gaussian distribution from another multivariate Gaussian distribution is well-known, e.g., \cite{duchi2020derivations}. 
\begin{lemma}
\lemlab{lem:kl:multivariates}
For $P=\calN(\mu_1,\Sigma_1)$ and $Q=\calN(\mu_2,\Sigma_2)$, we have
\[\KLD(P||Q)=\frac{1}{2}\left(\log\frac{\det(\Sigma_2)}{\det(\Sigma_1)}-r+\Trace(\Sigma_2^{-1}\Sigma_1)+(\mu_2-\mu_1)^\top\Sigma_2^{-1}(\mu_2-\mu_1)\right).\]
\end{lemma}
Recall the following relationship between total variation distance and KL divergence:
\begin{theorem}[Pinsker's inequality]
\thmlab{thm:pinsker}
For probability distributions $P$ and $Q$, we have
\[\TVD(P,Q)\le\sqrt{\frac{1}{2}\KLD(P||Q)}.\]
\end{theorem}
We now upper bound the total variation distance of the image of $\bA$ after right-multiplying with two vectors $\bx^{(1)}$ and $\bx^{(2)}$ whose entries are drawn from a normal distribution and a sparse scaled normal distribution, respectively. 
\begin{lemma}
\lemlab{lem:tvd:p}
Let $\gamma\ge 1$ be any fixed constant and let $s=\Theta\left(\frac{\gamma^2}{p^2}\cdot r^4\log r\right)$. 
Let $D=\calN(0,1)$ and let $D_p=\Bern\left(p\right)\cdot\calN\left(0,\frac{1}{p}\right)$ for a constant $p \in (0, 1)$. 
Let $\bx^{(1)}\sim D$ and $\bx^{(2)}\sim D_p$. 
Let $\bA\in\mathbb{R}^{r\times n}$ be a matrix with leverage score at most $\frac{1}{s}$.
Then $\TVD(\bA\bx^{(1)},\bA\bx^{(2)})\le\O{\frac{1}{\gamma }}$. 
\end{lemma}
\begin{proof}
Since $\bx^{(1)}$ is a multivariate Gaussian with identity covariance matrix and mean $0^n$, then $\bA\bx^{(1)}$ is a multivariate Gaussian with mean $0^r$ and covariance matrix $\bA\bA^\top$. 
Let $\bx^{(2)}\sim\calD_2$ and let $S$ be the support of $\bx^{(2)}$. 
Note that each coordinate of $\bx^{(2)}$ in the support of $S$ is drawn from the Gaussian distribution $\calN(0,\frac{1}{p})$. 
Therefore, $\bA\bx^{(2)}$ is a multivariate Gaussian with mean $0^r$ and covariance matrix $\bB\bB^\top$ for some matrix $\bB$. 
By \lemref{lem:kl:multivariates}, we have that 
\[\KLD(\bA\bx^{(1)},\bA\bx^{(2)})=\frac{1}{2}\left(\log\frac{\det(\bB^\top\bB)}{\det(\bA^\top\bA)}-r+\Trace((\bB^\top\bB)^{-1}(\bA^\top\bA)\right).\]
Let $\calE$ be the event that $\left(1-\frac{1}{\gamma r}\right)^2\bB^\top\bB\preceq\bA^\top\bA\preceq\left(1+\frac{1}{\gamma r}\right)^2\bB^\top\bB$, so that by \lemref{lem:spectral:preserve}, $\PPr{\calE}\ge 1-\frac{1}{\poly(r)}$. 
Then we have conditioned on $\calE$, 
\[\KLD(\bA\bx^{(1)},\bA\bx^{(2)}\,\mid\,\calE)\le\frac{1}{2}\left(r\cdot\log\left(1+\frac{1}{\gamma r}\right)-r+r\cdot\left(1+\frac{1}{\gamma r}\right)\right)=\O{\frac{1}{\gamma }},\]
since $\log(1+x)=\O{x}$ for $x\in\left(0,\frac{1}{2}\right)$, e.g., by the Taylor series expansion of $\log(1+x)$. 
By \thmref{thm:pinsker}, we thus have
\[\TVD(\bA\bx^{(1)},\bA\bx^{(2)}\,\mid\,\calE)\le\O{\frac{1}{\gamma }}.\]
Since $\PPr{\calE}\ge 1-\frac{1}{\poly(r)}$, then it follows that
\[\TVD(\bA\bx^{(1)},\bA\bx^{(2)})\le\O{\frac{1}{\gamma }}.\]
\end{proof}

Combining \lemref{lem:tvd:p} and the triangle inequality, we have the following lemma immediately. 
\begin{lemma}
\lemlab{lem:tvd:pq}
Let $\gamma\ge 1$ be any fixed constant and let $s=\Theta(\gamma^2 r^4\log r)$. 
Let $D_p=\Bern\left(p\right)\cdot\calN\left(0,\frac{1}{p}\right)$ and $D_q=\Bern\left(p\right)\cdot\calN\left(0,\frac{1}{p}\right)$ for constant $p, q \in (0, 1)$. 
Let $\bx^{(1)}\sim D_p$ and $\bx^{(2)}\sim D_q$. 
Let $\bA\in\mathbb{R}^{r\times n}$ be a matrix with leverage score at most $\frac{1}{s}$.
Then $\TVD(\bA\bx^{(1)},\bA\bx^{(2)})\le\O{\frac{1}{\gamma }}$. 
\end{lemma}

We are now ready to present our attack over real-valued inputs, which is shown in Figure~\ref{fig:2}. Note that this is analogous to the attack for $\mathbb{Z}^{r \times n}$, and the only difference is that we use a different input distribution and a different setting of parameters $h, \sigma$.

\begin{figure}

\begin{mdframed}
\begin{algorithmic}

Let $\alpha$ and $\beta$ be two constants such that $\alpha$ is close to $0$ and $\beta$ is close to $1$.
\newline
Let $\mathcal{D}$ be the distribution family where $D_p = \Bern\left(p\right)\cdot\calN(0,\frac{1}{p})$
\newline\noindent
$h \gets\O{r^2 s \log r} = \O{r^{12} \log r}$, $\sigma\gets\O{h \log(n)}$, $\ell\gets\O{h}\cdot\sigma$, $c \gets \O{1}$
\newline \noindent
Let $z_J(v) $ denote the vector where we make $v_i$  to $0$ for all $i \in J$.
\newline$\mathcal{A} \gets$ An instantiation of the $\ell_0$ gap-norm algorithm. %-
\newline\noindent
Initialize $s_i^0=0$ for all $i\in[n]$
\newline \noindent 
For $j\in[\ell]$:
\newline \indent
Sample $u^1, \cdots, u^c \sim D_\alpha^n$ and $v^1, \cdots, v^c \sim D_{\beta}^n$ % v^1, \cdots v^c 
% \newline \indent
% Set $u^i \gets u^i_{I^j}$ and $v^i \gets v^i_{I^j}$ where we zero the coordinates on $I^j$% let
\newline \indent
If $\mathcal{A}$ fails with constant probability on one of $z_{I^{j - 1}}(u^i)$ or $z_{I^{j - 1}}(v^i)$:
\newline \indent \indent Output this distribution as the attack.
\newline\indent
Sample $p^j\sim P_{\alpha,\beta}$ and $v^j\sim D_{p^{j}}^n$
% \newline \indent
% Set $v^i \gets v^i_{I^j}$.
\newline\indent
For all $i\in[n]$, set $c_i^j=1$ if $v_i^j\neq 0$ and $c_i^j=-1$ otherwise if $v_i^j=0$
\newline\indent
Query $z_{I^{j - 1}}(v^j)\in\mathbb{R}^n$ and receive $a^j = \mathcal{A}(z_{I^{j - 1}}(v^j))\in\{\pm1\}$ as the output % v^j_{[n] \setminus I^{j - 1}}
\newline\indent
For $i\in[n]$, update $s_i^j\gets s_i^{j-1}+a^j\cdot\phi^{p^j}(c_i^j)$
\newline\indent
Set $I^j= I^{j - 1} \cup \{i\in[n]\,\mid\, s_i^j>\sigma\}$ and $\mathcal{S}^{j + 1} = \mathcal{S} \setminus I^j$    
\end{algorithmic}
\end{mdframed}   
\caption{Construction of Our Attack over the Reals}
\label{fig:2}
\end{figure} 

We now prove \thmref{thm:real}. 
\thmreal*
\begin{proof}
Our argument is similar to that of Section~\ref{sec:attack}. Recall that we set $\gamma = r^3$ in \lemref{lem:tvd:pq}, which means that $s = \O{r^{10} \log r}$ and $r^2s = \O{r^{12} \log r}$ and hence, we can assume the sketching matrix $\bA$ has the form
\[
\bA = \begin{bmatrix}
        \bD \\
        \bS
    \end{bmatrix} \;,
\]
where $\bS$ has at most $r^2 s$ non-zero columns and $\bD$ satisfies 
 \begin{equation}
 \label{eq:real}
    \forall \by^\top\in \mathbb{R}^r, (\by^\top \bD)_i^2 \le \frac{1}{s}\cdot \|\by^\top \bD\|_2^2 \; .
\end{equation}
Let $\mathcal{S}$ denote the indices of the at most $r^2 s$ non-zero columns.

\paragraph{Soundness.} 
Consider the coordinates $i \in I \setminus \mathcal{S}$. 
From the choice of the parameters we have the total variation distance of $\bD  \bx^{D}$ for different $p \in [\alpha, \beta]$ is $\O{\frac{1}{r^3}}$ where $\bx \sim D_p$. 
Then, from this we can get that there are at most $\O{r^9}$ coordinates $i \in I \setminus \mathcal{S}$ such that the expectation of $s_i^{t} - s_i^{t - 1}$ is $\Omega\left(\frac{1}{r^{12}}\right)$ given $\bD  \bx^{D}$, which means that with high probability there are at most $\tO{r^9} = o(s)$ coordinates in $I \setminus \mathcal{S}$ will be accused in the procedure (as there are total $\tO{r^2 s^4}$ number of queries).

Suppose that $\bD$ is the matrix that satisfies $\eqref{eq:real}$ and $\bD'$ is the matrix where we zero out $o(s)$ columns of $\bD$. Then for any $\by^\top \in \mathbb{R}^{r}$, since for each remaining index $i$ satisfies $(\by^\top \bD)_i^2\le\frac{1}{s}\cdot \|\by^\top \bD\|_2^2 $, then
\[
    \forall \by^\top\in \mathbb{R}^r, (\by^\top \bD')_i^2 \le \frac{1.1}{s}\cdot \|\by^\top \bD'\|_2^2 \; .
\]
% Note that from the condition we consider in~\eqref{eq:real}, even if we made these false accusation, the condition for the columns in $D$ still holds and it will only increase the total variation distance of $\bD \bx^{D}$ for different $p$ by at most a small constant factor.

Hence, in the rest of the argument, we can assume the total variation bound for $\bD'  \bx^{D}$ for different $p$ still holds.

\paragraph{Completeness.}
Suppose that the algorithm $\mathcal{A}$ we attack uses the estimator $f$. 
We consider $\mathcal{A}'$ to be an algorithm that uses the same estimator $f$, but rather than $f$ takes the sketch $\bA \bx$ as the input, it takes the input $\begin{bmatrix}
\bD' \bx'\\
\bS \bx_S
\end{bmatrix}$
where $\bx' \sim D_{\gamma}^{|D|}$ for a fixed $\gamma \in [\alpha, \beta]$, which is sampled by the algorithm $\mathcal{A}$ and is independent of the input $\bx$. 
From \lemref{lem:distribution:tvd}, we know that for each iteration $t$, the total variation distance between $\begin{bmatrix}
        \bD' \bx^{(t)}_{D} \\
        \bS \bx^{(t)}_{\mathcal{S}}
    \end{bmatrix}$
    and 
    $\begin{bmatrix}
        \bD' \bx' \\
        \bS \bx^{(t)}_{\mathcal{S}}
    \end{bmatrix}$
is at most $\O{\frac{1}{r^3}}$ for some small constant $\gamma$. If $\mathcal{A}$ succeeds with probability at least $1 - \delta$ over some input distribution $D$, then over this distribution $\mathcal{A}$ will also succeed with probability at least $1 - \delta - \O{\frac{1}{r^3}}$. 
    
    % Let us first assume the algorithm $\mathcal{A}$ uses estimator $f'$ which only considers the input of $\bS\bx_\mathcal{S}$ and then randomly generates an input $\bx_D \sim D_{\gamma}^{|D|}$ to compute a $\bD \bx_{D}$ as another part of the input vector, where $\gamma \in [\alpha, \beta]$ is some fixed constant.

Let us now first assume the algorithm we attack is $\mathcal{A}'$. 
Note that $\mathcal{A'}$ has the property that it only uses $x_{\mathcal{S}}$ in the computation. 
From \lemref{soundness}, we see that with probability at least $1-\frac{1}{n}$, we never falsely accuse any index $i \notin \mathcal{S}$. 
    % this implies that, except with probability $\frac{1}{n}$, we will never zero-out a coordinate that does not belong to the sparse part $\bS$ of $\bA$. 
Additionally, by \lemref{lem:completeness}, we know that with high constant probability, our attack correctly identifies (some, or all) coordinates $i \in S$ and outputs a distribution on which $\mathcal{A}'$ fails (note that the increase of the $\O{\frac{1}{r^3}}$ in the error probability will make the $g(\beta) - g(\alpha)$ in~\lemref{lem:function:gap} decrease by at most $\O{\frac{1}{r^3}}$, which is still $\Omega(1)$). 
It follows that with high constant probability, our attack finds some hard query distribution $\bq$ on which $\mathcal{A'}$ fails with constant probability. 
Now, let us now consider the original algorithm $\mathcal{A}$. For the property of the total variation distance we know that with probability at least $1 - \frac{1}{r^3}$, the output of $\mathcal{A}$ can be seen as sampled from the same distribution from the output of $\mathcal{A}'$, and the probability is only over the random choice over input vector $\bx$. 
Hence, without loss of generality it can be seen as the algorithm $\mathcal{A}'$ but with a $\O{\frac{1}{r^3}}$ more inconsistency, and the probability here is only over the choice of the random input vector $\bx$ (note that the algorithm is based on a linear sketch and the output of the algorithm is binary). 
Therefore, the same argument still applies.
% We next consider the query complexity. 

\end{proof}

\section*{Acknowledgements}
Elena Gribelyuk and Huacheng Yu are supported in part by an NSF CAREER award CCF-2339942. 
Honghao Lin was supported in part by a Simons Investigator Award, NSF CCF-2335412, and a CMU Paul and James Wang Sercomm Presidential Graduate Fellowship. 
David P. Woodruff was supported in part by a Simons Investigator Award and NSF CCF-2335412. 
Samson Zhou is supported in part by NSF CCF-2335411. 
The authors would like to acknowledge the Sketching and Algorithm Design workshop held at the Simons Institute for the Theory of Computing for contributions to the development of this work. 

\def\shortbib{0}
\bibliographystyle{alpha}
\bibliography{references}
\appendix
 
\end{document}